\documentclass[a4paper,british]{article}

\usepackage[utf8]{inputenc}
\usepackage[T1]{fontenc}
\usepackage{newunicodechar}
\usepackage[only,llbracket,rrbracket]{stmaryrd}
\usepackage{amsmath,mathtools}
\usepackage{tikz-cd}
\usepackage{todonotes}
\usepackage{tcolorbox}
\usepackage[style=alphabetic,backend=biber,maxbibnames=99,maxcitenames=2]{biblatex}
\usepackage{hyperref}
\hypersetup{
  colorlinks=true,
  linkcolor=blue,
  urlcolor=red
}
\usepackage{quiver}
\usepackage{amsthm}
\usepackage[final]{microtype}

\usepackage[
, textwidth = 140mm
, textheight = 222mm
, margin = 38mm
, headheight = 3mm
, headsep = 10mm
, footskip = 5mm
]{geometry}

\usepackage[frozencache=true]{minted}

\newunicodechar{ₓ}{\ensuremath{_{\texttt{x}}}}
\newunicodechar{⟶}{\ensuremath{\longrightarrow}}
\newunicodechar{≫}{\ensuremath{\gg}}
\newunicodechar{⋙}{\ensuremath{\ggg}}
\newunicodechar{⟦}{\ensuremath{\llbracket}}
\newunicodechar{⟧}{\ensuremath{\rrbracket}}
\newunicodechar{𝟙}{\ensuremath{\bf1}}

\newunicodechar{α}{{\ensuremath{\mathrm{\alpha}}}}
\newunicodechar{β}{{\ensuremath{\mathrm{\beta}}}}
\newunicodechar{γ}{{\ensuremath{\mathrm{\gamma}}}}
\newunicodechar{δ}{{\ensuremath{\mathrm{\delta}}}}
\newunicodechar{ε}{{\ensuremath{\mathrm{\varepsilon}}}}
\newunicodechar{ζ}{{\ensuremath{\mathrm{\zeta}}}}
\newunicodechar{η}{{\ensuremath{\mathrm{\eta}}}}
\newunicodechar{θ}{{\ensuremath{\mathrm{\theta}}}}
\newunicodechar{ι}{{\ensuremath{\mathrm{\iota}}}}
\newunicodechar{κ}{{\ensuremath{\mathrm{\kappa}}}}
\newunicodechar{μ}{{\ensuremath{\mathrm{\mu}}}}
\newunicodechar{ν}{{\ensuremath{\mathrm{\nu}}}}
\newunicodechar{ξ}{{\ensuremath{\mathrm{\xi}}}}
\newunicodechar{π}{{\ensuremath{\mathrm{\mathnormal{\pi}}}}}
\newunicodechar{ρ}{{\ensuremath{\mathrm{\rho}}}}
\newunicodechar{σ}{{\ensuremath{\mathrm{\sigma}}}}
\newunicodechar{τ}{{\ensuremath{\mathrm{\tau}}}}
\newunicodechar{φ}{{\ensuremath{\mathrm{\varphi}}}}
\newunicodechar{χ}{{\ensuremath{\mathrm{\chi}}}}
\newunicodechar{ψ}{{\ensuremath{\mathrm{\psi}}}}
\newunicodechar{ω}{{\ensuremath{\mathrm{\omega}}}}

\newunicodechar{Γ}{{\ensuremath{\mathrm{\Gamma}}}}
\newunicodechar{Δ}{{\ensuremath{\mathrm{\Delta}}}}
\newunicodechar{Θ}{{\ensuremath{\mathrm{\Theta}}}}
\newunicodechar{Λ}{{\ensuremath{\mathrm{\Lambda}}}}
\newunicodechar{Σ}{{\ensuremath{\mathrm{\Sigma}}}}
\newunicodechar{Φ}{{\ensuremath{\mathrm{\Phi}}}}
\newunicodechar{Ξ}{{\ensuremath{\mathrm{\Xi}}}}
\newunicodechar{Ψ}{{\ensuremath{\mathrm{\Psi}}}}
\newunicodechar{Ω}{{\ensuremath{\mathrm{\Omega}}}}
\newunicodechar{ℵ}{{\ensuremath{\aleph}}}
\newunicodechar{≤}{{\ensuremath{\leq}}}
\newunicodechar{≥}{{\ensuremath{\geq}}}
\newunicodechar{≠}{{\ensuremath{\neq}}}
\newunicodechar{≈}{{\ensuremath{\approx}}}
\newunicodechar{≡}{{\ensuremath{\equiv}}}
\newunicodechar{≃}{{\ensuremath{\simeq}}}
\newunicodechar{∂}{{\ensuremath{\partial}}}
\newunicodechar{∆}{{\ensuremath{\triangle}}} 

\newunicodechar{∫}{{\ensuremath{\int}}}
\newunicodechar{∑}{{\ensuremath{\mathrm{\Sigma}}}}
\newunicodechar{Π}{{\ensuremath{\mathrm{\Pi}}}}
\newunicodechar{⊥}{{\ensuremath{\perp}}}
\newunicodechar{∞}{{\ensuremath{\infty}}}
\newunicodechar{∓}{{\ensuremath{\mp}}}
\newunicodechar{⊕}{{\ensuremath{\oplus}}}
\newunicodechar{⊗}{{\ensuremath{\otimes}}}
\newunicodechar{⊞}{{\ensuremath{\boxplus}}}
\newunicodechar{∇}{{\ensuremath{\nabla}}}
\newunicodechar{√}{{\ensuremath{\sqrt}}}
\newunicodechar{⬝}{{\ensuremath{\cdot}}}
\newunicodechar{∘}{{\ensuremath{\circ}}}

\newunicodechar{⁻}{{\ensuremath{^{-}}}}
\newunicodechar{▸}{{\ensuremath{\blacktriangleright}}}
\newunicodechar{∧}{{\ensuremath{\wedge}}}
\newunicodechar{∨}{{\ensuremath{\vee}}}
\newunicodechar{⊢}{{\ensuremath{\vdash}}}
\newunicodechar{↦}{{\ensuremath{\mapsto}}}
\newunicodechar{↔}{{\ensuremath{\leftrightarrow}}}
\newunicodechar{⇒}{{\ensuremath{\Rightarrow}}}
\newunicodechar{⟹}{{\ensuremath{\Longrightarrow}}}
\newunicodechar{⇐}{{\ensuremath{\Leftarrow}}}
\newunicodechar{⟸}{{\ensuremath{\Longleftarrow}}}
\newunicodechar{∩}{{\ensuremath{\cap}}}
\newunicodechar{∪}{{\ensuremath{\cup}}}
\newunicodechar{⊂}{{\ensuremath{\subseteq}}}
\newunicodechar{⊆}{{\ensuremath{\subseteq}}}
\newunicodechar{⊄}{{\ensuremath{\nsubseteq}}}
\newunicodechar{⊈}{{\ensuremath{\nsubseteq}}}
\newunicodechar{⊃}{{\ensuremath{\supseteq}}}
\newunicodechar{⊇}{{\ensuremath{\supseteq}}}
\newunicodechar{⊅}{{\ensuremath{\nsupseteq}}}
\newunicodechar{⊉}{{\ensuremath{\nsupseteq}}}
\newunicodechar{∈}{{\ensuremath{\in}}}
\newunicodechar{∉}{{\ensuremath{\notin}}}
\newunicodechar{∋}{{\ensuremath{\ni}}}
\newunicodechar{∌}{{\ensuremath{\notni}}}
\newunicodechar{∅}{{\ensuremath{\emptyset}}}
\newunicodechar{∖}{{\ensuremath{\setminus}}}
\newunicodechar{ℕ}{{\ensuremath{\mathbb{N}}}}
\newunicodechar{ℤ}{{\ensuremath{\mathbb{Z}}}}
\newunicodechar{ℝ}{{\ensuremath{\mathbb{R}}}}
\newunicodechar{ℚ}{{\ensuremath{\mathbb{Q}}}}
\newunicodechar{ℂ}{{\ensuremath{\mathbb{C}}}}
\newunicodechar{⌞}{{\ensuremath{\llcorner}}}
\newunicodechar{⌟}{{\ensuremath{\lrcorner}}}
\newunicodechar{⦃}{{\ensuremath{\{\!|}}}
\newunicodechar{⦄}{{\ensuremath{|\!\}}}}
\newunicodechar{₁}{{\ensuremath{_1}}}
\newunicodechar{₂}{{\ensuremath{_2}}}
\newunicodechar{₃}{{\ensuremath{_3}}}
\newunicodechar{₄}{{\ensuremath{_4}}}
\newunicodechar{₅}{{\ensuremath{_5}}}
\newunicodechar{₆}{{\ensuremath{_6}}}
\newunicodechar{₇}{{\ensuremath{_7}}}
\newunicodechar{₈}{{\ensuremath{_8}}}
\newunicodechar{₉}{{\ensuremath{_9}}}
\newunicodechar{₀}{{\ensuremath{_0}}}
\newunicodechar{ᵢ}{{\ensuremath{_i}}}
\newunicodechar{ⱼ}{{\ensuremath{_j}}}
\newunicodechar{ₐ}{{\ensuremath{_a}}}
\newunicodechar{ₙ}{{\ensuremath{_n}}}
\newunicodechar{ₘ}{{\ensuremath{_m}}}
\newunicodechar{∀}{{\color{symbolcolor}\ensuremath{\forall}}}
\newunicodechar{∃}{{\color{symbolcolor}\ensuremath{\exists}}}
\newunicodechar{λ}{{\color{symbolcolor}\ensuremath{\mathrm{\lambda}}}}

\usepackage[capitalise]{cleveref}

\makeatletter
\ifwindows
  \renewcommand{\minted@optlistcl@quote}[2]{%
    \ifstrempty{#2}{\detokenize{#1}}{\detokenize{#1="#2"}}}
\else
  \renewcommand{\minted@optlistcl@quote}[2]{%
    \ifstrempty{#2}{\detokenize{#1}}{\detokenize{#1='#2'}}}
\fi

\newcommand{\minted@def@optcl@novalue}[2]{%
  \define@booleankey{minted@opt@g}{#1}%
    {\minted@addto@optlistcl{\minted@optlistcl@g}{#2}{}%
     \@namedef{minted@opt@g:#1}{true}}
    {\@namedef{minted@opt@g:#1}{false}}
  \define@booleankey{minted@opt@g@i}{#1}%
    {\minted@addto@optlistcl{\minted@optlistcl@g@i}{#2}{}%
     \@namedef{minted@opt@g@i:#1}{true}}
    {\@namedef{minted@opt@g@i:#1}{false}}
  \define@booleankey{minted@opt@lang}{#1}%
    {\minted@addto@optlistcl@lang{minted@optlistcl@lang\minted@lang}{#2}{}%
     \@namedef{minted@opt@lang\minted@lang:#1}{true}}
    {\@namedef{minted@opt@lang\minted@lang:#1}{false}}
  \define@booleankey{minted@opt@lang@i}{#1}%
    {\minted@addto@optlistcl@lang{minted@optlistcl@lang\minted@lang @i}{#2}{}%
     \@namedef{minted@opt@lang\minted@lang @i:#1}{true}}
    {\@namedef{minted@opt@lang\minted@lang @i:#1}{false}}
  \define@booleankey{minted@opt@cmd}{#1}%
      {\minted@addto@optlistcl{\minted@optlistcl@cmd}{#2}{}%
        \@namedef{minted@opt@cmd:#1}{true}}
      {\@namedef{minted@opt@cmd:#1}{false}}
}

\minted@def@optcl@novalue{custom lexer}{-x}
\makeatother

\newmintinline[lean]{lean4.py:Lean4Lexer}{custom lexer, bgcolor=white}
\newminted[leancode]{lean4.py:Lean4Lexer}{custom lexer, fontsize=\footnotesize}
\makeatletter
\AtBeginEnvironment{minted}{\dontdofcolorbox}
\def\dontdofcolorbox{\renewcommand\fcolorbox[4][]{##4}}
\makeatother

\usepackage{general-macros}
\usepackage{local-macros}
\usepackage{thmstyle}

\newcommand{\leftsquigarrow}{\mathrel{\reflectbox{$\rightsquigarrow$}}}

\newtcolorbox[auto counter,
crefname={diagram}{diagrams},
Crefname={Diagram}{Diagrams}]{diagrambox}[2][]{%
colback=white,colframe=purple!75!black,fonttitle=\bfseries,
title=Diagram~\thetcbcounter: #2,#1}

\NewDocumentEnvironment{diagram}{o m m}{
  \IfNoValueTF{#1}{
    \begin{diagrambox}[float,label=#2]{#3}
  }{
    \begin{diagrambox}[float,floatplacement=#1,label=#2]{#3}
  }
}{
  \end{diagrambox}
}

\AtEveryBibitem{%
  \ifentrytype{inproceedings}{\iffieldundef{doi}{}{%
      \clearfield{url}%
    }}{}%
  \ifentrytype{article}{\iffieldundef{doi}{}{%
      \clearfield{url}%
    }}{}%
  \ifentrytype{online}{\iffieldundef{eprint}{}{%
      \clearfield{url}%
    }}{}%
}

\title{The Directed Van Kampen Theorem in Lean}
\author{
Henning Basold\thanks{LIACS, Leiden University, \url{mailto:h.basold@liacs.leidenuniv.nl}} \and
Peter Bruin\thanks{Mathematical Institute, Leiden University,
\url{mailto:p.j.bruin@math.leidenuniv.nl}.
Partially supported by the Dutch Research Council (NWO/OCW),
as part of the Quantum Software Consortium programme
(project number 024.003.037).} \and
Dominique Lawson\thanks{Student, Leiden University, \url{mailto:d.r.lawson@umail.leidenuniv.nl}.
  Based on the author's bachelor thesis.}
}
\date{Pre-print}

\begin{document}
\maketitle

\begin{abstract}
  Directed topology is an area of mathematics with applications in concurrency.
  It extends the concept of a topological space by adding a notion of directedness, which restricts
  how paths can evolve through a space and enables thereby a faithful representation of computation
  with their direction.
  In this paper, we present a Lean formalisation of directed spaces and a Van Kampen theorem for
  them.
  This theorem allows the calculation of the homotopy type of a space by combining local knowledge
  the homotopy type of subspaces.
  With this theorem, the reasoning about spaces can be reduced to subspaces and, by representing
  concurrent systems as directed spaces, we can reduce the deduction of properties of a composed
  system to that of subsystems.
  The formalisation in Lean can serve to support computer-assisted reasoning about the behaviour
  of concurrent systems.
\end{abstract}

\section{Introduction}
\label{introduction}

The direction of paths that may pass through a topological space, and thus the behaviour of
dynamical systems in it, are only constrained by the shape of the space.
However, there are often cases when paths must follow a particular direction and are generally
not reversible.
One motivation, stems from models of
true concurrency~\cite{Fajstrup98:DetectingDeadlocksConcurrent}, where executions are
modelled as non-reversible paths in a space.
For instance, two programs A and~B can be executed \emph{sequentially} in two ways: either we
first run A and then B, or vice versa, see \cref{fig:directed-situations} on the left.
This choice between two sequential linearisations corresponds to semantics of labelled transition
systems, but it neglects potential parallel execution.
To see this, suppose that A and B have no dependency or interaction and can be run in parallel.
This situation can be modelled by admitting any path in the square from the bottom left to the top
right as a valid execution, with the intuition that going along the path tracks how far each of the
processes has been run, see b) of \cref{fig:directed-situations}.
The caveat is that processes can, in general, not be reversed and therefore the path may only ever
go up and to the right, thereby following the directions of the arrows.
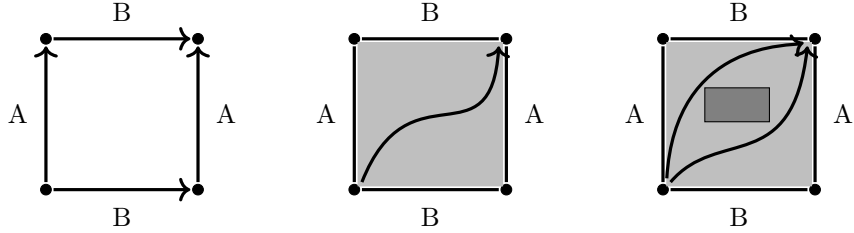
\begin{figure}[ht]
  \centering
  \begin{tikzpicture}
    \filldraw [black] (0,0) circle (2pt);
    \filldraw [black] (0,2) circle (2pt);
    \filldraw [black] (2,0) circle (2pt);
    \filldraw [black] (2,2) circle (2pt);
    \draw[very thick, ->] (0,0.1) -- node[left=1mm] {A} (0, 1.9);
    \draw[very thick, ->] (0.1,0) -- node[below=1mm]{B} (1.9, 0);
    \draw[very thick, ->] (2,0.1) -- node[right=1mm]{A} (2, 1.9);
    \draw[very thick, ->] (0.1,2) -- node[above=1mm]{B} (1.9, 2);
  \end{tikzpicture}
  \qquad
  \begin{tikzpicture}
    \filldraw [black] (0,0) circle (2pt);
    \filldraw [black] (0,2) circle (2pt);
    \filldraw [black] (2,0) circle (2pt);
    \filldraw [black] (2,2) circle (2pt);
    \filldraw [lightgray] (0.05, 0.05) rectangle (1.95, 1.95);
    \draw[very thick] (0,0.1) -- node[left=1mm]  {A} (0, 1.9);
    \draw[very thick] (0.1,0) -- node[below=1mm] {B} (1.9, 0);
    \draw[very thick] (2,0.1) -- node[right=1mm] {A} (2, 1.9);
    \draw[very thick] (0.1,2) -- node[above=1mm] {B} (1.9, 2);
    \draw[very thick, ->] (0.1, 0.1) .. controls (0.72, 1.68) and (1.8, 0.32) .. (1.9, 1.9);
  \end{tikzpicture}
  \qquad
  \begin{tikzpicture}
    \filldraw [black] (0,0) circle (2pt);
    \filldraw [black] (0,2) circle (2pt);
    \filldraw [black] (2,0) circle (2pt);
    \filldraw [black] (2,2) circle (2pt);
    \filldraw [lightgray] (0.05, 0.05) rectangle (1.95, 1.95);
    \draw [fill=gray] (0.55, 0.9) rectangle (1.4, 1.35);
    \draw[very thick] (0,0.1) -- node[left=1mm]  {A} (0, 1.9);
    \draw[very thick] (0.1,0) -- node[below=1mm] {B} (1.9, 0);
    \draw[very thick] (2,0.1) -- node[right=1mm] {A} (2, 1.9);
    \draw[very thick] (0.1,2) -- node[above=1mm] {B} (1.9, 2);
    \draw[very thick, ->] (0.1, 0.1) .. controls (0.72, 0.84) and (1.7, 0.16) .. (1.9, 1.9);
    \draw[very thick, ->] (0.05, 0.15) .. controls (0.16, 2) and (1.7, 1.9) .. (1.85, 1.95);
  \end{tikzpicture}
  \caption{Possible execution paths of two programs A and B under three conditions:
    a)~sequential (left), b) simultaneous (middle) and c) simultaneous with obstacles (right).}
  \label{fig:directed-situations}
\end{figure}
Now suppose that there is a dependency between the processes like, for instance, they need to write
to the same memory location.
In order to prevent race conditions, we can rule out those execution paths in which the processes
access the memory location at the same time.
This, in turn, can be modelled by the space square in c) of \cref{fig:directed-situations},
where the darker rectangle is an obstacle that paths have to bypass.
The two displayed paths in that space represent different access patterns to the memory, in that
the lower path indicates that process B first gets access to the memory location, while the upper
means that A first gets access.
These two paths are essentially different because the observable behaviour of the system differs and
because we cannot change the access pattern during execution.
In contrast, the different paths in the space b) of \cref{fig:directed-situations} model executions
that differ only in the relative execution speeds of A and B but are otherwise equivalent.
By giving one process more execution time, we can always deform one path into another in this space.
Finally, the space in a) has exactly two paths from the bottom left to top right, neither of which can
be deformed to the other during execution due to the absence of parallelism.
This tells us that the spaces in \cref{fig:directed-situations} all model different systems and
the question is then how our intuition about relating execution paths can be made
precise and how we can reason about these relations.

Directed topology and directed homotopy theory~\cite{Fajstrup16:DirectedAlgebraicTopology,%
  Grandis09:DirectedAlgebraicTopology}
make precise the above intuition and enable the analysis of concurrent systems with the tools
of algebraic topology.
There are various ways to enforce direction in topological spaces, such as
higher-dimensional automata~\cite{Glabbeek06:ExpressivenessHDA,%
  Pratt91:ModelingConcurrencyGeometry},
spaces with a global order~\cite{fajstrup2006},
spaces with local orders~\cite{Fajstrup03:DicoveringSpaces},
streams~\cite{Krishnan09:ConvenientCategoryLocally},
and various others~\cite{Dubut17:DirectedHomotopyHomology, Gaucher21:SixModelCategories}.
We will focus here on the notion of d-space~\cite{grandis2003directed}, which represents a
directed space as a topological space with a distinguished set of directed paths.
It then turns out that reasoning about concurrent systems becomes reasoning about the homotopy
type of d-spaces, that is, the relation between directed paths in a d-space.

One of the important steps in building and analysing large systems is to prove local properties of
subsystems and deduce properties of the whole composed system from these local properties.
In algebraic topology, one of the important results to combine knowledge of the homotopy type
of subspaces into knowledge about the whole space is the so-called Van Kampen
theorem~\cite{brown2006Topology}.
This result has been extended to d-spaces by \textcite{grandis2003directed}.
To make the use of this result applicable in larger systems, we set out in this paper to formalise
the van Kampen theorem for d-spaces in the proof assistant Lean~\cite{dKA+15:LeanTheoremProver},
thereby enabling compositional reasoning about the homotopy type of d-spaces and of concurrent
systems modelled as d-spaces.

\subsection{Contributions}

Our main contribution is the formalization of definitions and theorems relating to directed topology,
in particular the Van Kampen Theorem. For this formalization we used Lean 3.50.3 and we built upon
the work already present in MathLib~\cite{ThemathlibCommunity20:LeanMathematicalLibrary}.
All of the formalization can be found in the accompanying Git
repository~\cite{Lawson23:GitHubDominiqueLawsonDirectedTopology}.

We also analyze the conditions of the Van Kampen Theorem and exhibit an example showing that
the directed Van Kampen Theorem does not directly generalize to fundamental monoids.

\subsection{Overview}

In \cref{Ch Directed Spaces}, we define the notion of directed spaces and directed maps and give a few examples.
In \cref{Ch Directed Homotopies}, the definitions and some properties of directed homotopies
and directed path homotopies are given.
We use those to define relations on the set of directed paths between two points.
In \cref{Ch Fundamental Structures}, the equivalence classes of paths under these relations
are used to define the fundamental category and the fundamental monoid.
The Van Kampen Theorem is stated in \cref{Ch Van Kampen Theorem} and an application is given.
Finally, in \cref{Ch Conclusion} we reflect on the ideas presented in this article
and give some suggestions for further research.

Small excerpts from the Lean formalization can be found throughout the article. 
They can be recognized by the \verb|monospace| font used.
The surrounding text explains the ideas of the formalization and references the corresponding file in which the code can be found.

\section{Directed Spaces}\label{Ch Directed Spaces}

In this section, we will look at the basic structure of a directed space.
With directed maps as morphisms, the category of directed spaces \textbf{dTop} is obtained.

\subsection{Directed Spaces}

A directed space is a topological space with a distinguished set of paths,
whose elements are called directed paths.
This is analogous to the set of open sets in a topological spaces.
Similarly, the set of directed paths must satisfy some properties.
First, constant paths must be directed.
Secondly, if there are two directed paths that connect in an end and start point,
the concatenation of those two paths should again be directed.
Lastly, it should be possible to follow a part of a directed path at a different speed,
as long as the direction does not reverse anywhere.
This can be captured in the property that monotone subparametrizations of directed paths must also be directed paths.

\begin{definition}[Directed space]
    A directed space is a topological space $X$, together with a set of paths in $X$, denoted $P_X$.
    That set must satisfy the following three properties:
\begin{enumerate}
  \item For any point $x \in X$, we have that $0_x \in P_X$, where $0_x$ is the constant path in $x$.
  \item For any two paths $\gamma_1, \gamma_2 \in P_X$ with $\gamma_1(1) = \gamma_2(0)$,
            we have that $\gamma_1 \odot \gamma_2 \in P_X$.
  \item For any path $\gamma \in P_X$ and any continuous, monotone map $\varphi : [0,1] \to [0,1]$,
            we have that $\gamma \circ \varphi \in P_X$.
\end{enumerate}
The elements of $P_X$ are called \textit{directed paths} or \textit{dipaths}.
\end{definition} 

We will first consider some examples of directed spaces.

\begin{example}[Directed unit interval]\label{Directed unit interval}
    We can give the unit interval a rightward direction.
    This is done by taking
        $P_{[0,1]} = \{ \varphi : [0,1] \to [0,1] \mid \varphi \text{ continuous and monotone} \}$.
    We will denote this directed space with $I$.
    More generally, every (pre)ordered space can be given a set of directed paths this way.
\end{example}

\begin{example}[Directed unit circle]
    One of the ways the unit circle $S^1 = \{ z \in \mathbb{C} \mid |z| = 1 \}$
    can be made into a directed space is by only allowing paths that go monotonously counterclockwise.
    Specifically, we take the set of directed paths
        $P_{S^1} = \{t \mapsto \exp(i \varphi(t)) \mid \varphi : [0,1] \to \mathbb{R} \text{ continuous and monotone}\}$.
    This directed space will be denoted with $S^1_{+}$.
\end{example}

\begin{example}[Maximal directed space]
    Any topological space $X$ can be made into a directed space by taking $P_X$ as the set of all paths in $X$.
    We will call this the maximal directedness on $X$.
    This is also sometimes called the indiscrete or natural directedness.
\end{example}

\begin{example}[Minimal directed space]
    Any topological space $X$ can be made into a directed space by taking $P_X = \{ 0_x \mid x \in X \}$.
    In other words, only the constant paths are directed paths.
    We will call this the minimal directedness on $X$.
    This is also sometimes called the discrete directedness.
\end{example}

\begin{example}[Product of directed spaces]\label{Product Spaces}
    If $(X, P_X)$ and $(Y, P_Y)$ are two directed spaces,
    then the space $X \times Y$ with the product topology can be made into a directed space
    by letting $P_{X \times Y} = \{ t \mapsto (\gamma_1(t), \gamma_2(t)) \mid \gamma_1 \in P_X \text{ and } \gamma_2 \in P_Y \}$.
    As we will see in \cref{Directed Maps},
    with this set of directed paths both projection maps will be examples of directed maps and
    $(X \times Y, P_{X \times Y})$ becomes a product in a categorical sense.
\end{example}

\begin{example}[Induced directed space]\label{Induced Directed Spaces}
    Let $X$ be a topological space and $(Y, P_Y)$ a directed space.
    Let a continuous map $f: X \to Y$ be given.
    If $\gamma : [0,1] \to X$ is a path in $X$, then $f \circ \gamma : [0,1] \to Y$ is a path in $Y$.
    We can make $X$ into a directed space by taking
        $P_X = \{ \gamma \in C([0,1], X) \mid f \circ \gamma \in P_Y \}$.
    It is not hard to verify that this satisfies all the properties of a directed space.
    In the special case that $X$ is a subspace of $Y$ and $f$ is the inclusion map,
    we find that every subspace of a directed space can be given a natural directedness.
\end{example}

We formalized the notion of a directed space by extending the \verb|topological_space| class.
In our formalization, we do not explicitly use a set containing paths.
Rather, being a directed path is a property of a path itself,
analogous to how being open is a property of a set in the \verb|topological_space| class.
Paths in topological spaces have been implemented in MathLib in the file
\href{https://github.com/leanprover-community/mathlib/blob/97eab48/src/topology/path_connected.lean}{\texttt{topology/path\_connected.lean}}.
A path has type \verb|path x y|, where its starting point is \verb|x| and its endpoint is \verb|y|.
The definition of a directed space can be found in \leanref{directed\_space.lean}.

\vspace{3mm}
\begin{leancode}
class directed_space (α : Type u) extends topological_space α :=
  (is_dipath : ∀ {x y}, path x y → Prop)
  (is_dipath_constant : ∀ (x : α), is_dipath (path.refl x))
  (is_dipath_concat : ∀ {x y z} {γ₁ : path x y} {γ₂ : path y z},
        is_dipath γ₁ → is_dipath γ₂ → is_dipath (path.trans γ₁ γ₂))
  (is_dipath_reparam : ∀ {x y : α} {γ : path x y} {t₀ t₁ : I} {f : path t₀ t₁},
        monotone f → is_dipath γ → is_dipath (f.map γ.continuous_to_fun))
\end{leancode}
\vspace{3mm}

The term \verb|is_dipath| determines whether a path is a directed path or not.
The three other terms are exactly the three properties of a directed space.
\verb!path.refl x! is the constant path in a point \verb!x! and \verb!path.trans!
is used for the concatenation of paths.
MathLib only has support for reparametrizations of paths (meaning that the endpoints must remain the same),
but we want to also allow strict subparametrizations.
We do this by interpreting the subparametrization $f$ as a monotone path in $[0,1]$.
Then the path $\gamma \circ f$ can be obtained using \verb!path.map!,
where we interpret $\gamma$ as a continuous map.

In \verb|constructions.lean|, different instances of directed spaces can be found:
topological spaces with a preorder (\cref{Directed unit interval}),
products of directed spaces (\cref{Product Spaces}) and
induced directedness (\cref{Induced Directed Spaces}).

For brevity, we introduce a notation for the set of all paths between $x$ and $y$.

\begin{definition}
    If $X$ is a directed space and $x, y \in X$ points,
    we use the shorthand notation $P_X(x, y)$ for the set
    $\{ \gamma \in P_X \mid \gamma(0) = x \text{ and } \gamma(1) = y \}$.
\end{definition}

This definition can also be seen as a type for our formalization.
That is exactly how to interpret the structure \verb|dipath|, found in
\leanref{dipath.lean}:

\vspace{3mm}
\begin{leancode}
    variables {X : Type u} [directed_space X]
    structure dipath (x y : X) extends path x y :=
        (dipath_to_path : is_dipath to_path)
\end{leancode}
\vspace{3mm}

It extends the \verb|path| structure and depends on two points \verb|x| and \verb|y|
in a directed space \verb|X|.
The term \verb|dipath_to_path| has type \verb|is_dipath to_path|.
That means that the underlying path it extends must be a directed path.
Due to the axioms of a directed space,
we can define \verb|dipath.refl| and \verb|dipath.trans| analogous to their path-counterparts.
However, \verb|path.symm|, the reversal of a path,
cannot be converted to a directed variant
as it is not guaranteed that the reversal of a directed path is directed.

We introduce a notation for a special kind of subpath of a directed path.

\begin{definition}\label{Not path split part}
    Let $X$ be a directed space and $\gamma \in P_X$ a directed path.
    If $n > 0$ and $1 \le i \le n$ we will define $\gamma_{i, n} \in P_X$ to be the path
    from $\gamma(\tfrac{i-1}{n})$ to $\gamma(\tfrac{i}{n})$
    given by $\gamma_{i, n}(t) = \gamma(\tfrac{i+t-1}{n})$.
\end{definition}

We can now say what it means for a directed path to be covered by a cover of a directed space.
This definition will play a big role in proving and formalizing the Van Kampen Theorem.

\begin{definition}\label{Def path cover}
    Let $X$ be a directed space and $\mathcal{U}$ a cover of $X$.
    Let $\gamma \in P_X$ be a directed path and $n > 0$ an integer.
    We say that $\gamma$ is $n$-covered (by $\mathcal{U}$) if we have for all $1 \le i \le n$ that
    $\text{Im } \gamma_{i, n} \subseteq U$ for some $U \in \mathcal{U}$.
    In the special case that $n = 1$, we simply say that $\gamma$ is covered by $U$,
    where $\text{Im } \gamma \subseteq U$.
\end{definition}

In \leanref{path\_cover.lean} we formalize this definition of $n$-covered in the special case
that $\mathcal{U}$ consists of two elements $X_0$ and $X_1$:
\vspace{3mm}
\begin{leancode}
def covered (γ : dipath x₀ x₁) (hX : X₀ ∪ X₁ = univ) : Prop :=
    (range γ ⊆ X₀) ∨ (range γ ⊆ X₁)

def covered_partwise (hX : X₀ ∪ X₁ = set.univ) :
    Π {x y : X}, dipath x y → ℕ → Prop
| x y γ 0 := covered γ hX
| x y γ (nat.succ n) :=
    covered (split_dipath.first_part_dipath γ
        (inv_I_pos (show 0 < (n.succ + 1), by norm_num))) hX ∧
    covered_partwise (split_dipath.second_part_dipath γ
        (inv_I_lt_one (show 1 < (n.succ + 1), by norm_num))) n
\end{leancode}
\vspace{3mm}
Here \verb|covered| corresponds with $\gamma$ being $1$-covered:
its image is either contained in $X_0$ or in~$X_1$.
We use this definition to inductively define \verb|covered_partwise|.
As it is easier to start at zero in Lean,
\verb|covered_partwise hX γ n| corresponds with $\gamma$ being $(n+1)$-covered.
In the case that $n = 0$, we have that \verb|covered_partwise| simply agrees with \verb|covered hX γ|.
Otherwise, we use an induction step to define that \verb|covered_partwise hX γ (nat.succ n)| holds
if the first part $\gamma_{1, {n+2}}$ is \verb|covered|
and the remainder of $\gamma$ is \verb|covered_partwise hX γ n|.
Note the use of $n + 2$ instead of $n+1$ due to the offset between the definitions.
The remainder of \leanref{path\_cover.lean} contains lemmas about conditions for being $n$-covered.

\subsection{Directed Maps}\label{Directed Maps}

As directed spaces are an extension of topological spaces,
directed maps will be extensions of continuous maps.
They will need to respect the extra directed structure.
If a path in the domain space is given,
a path in the codomain space can be obtained by composing the continuous map with the path. 
If the former is directed, so should be the latter. 

\begin{definition}[Directed map] Let $X$ and $Y$ be two directed spaces.
    A directed map $f : X \to Y$ is a continuous map on the underlying topological spaces
    that furthermore satisfies: for any $\gamma \in P_X$, we have that $f \circ \gamma \in P_Y$.
\end{definition}

\begin{example}
    The map $f : I \to S^1_+$ given by $t \mapsto e^{it}$ is a directed map.
    First, it is continuous on the underlying
    topological spaces.
    Secondly, if $\gamma \in P_I$ is a directed path, that is, continuous and monotone,
    then $f \circ \gamma$ is given by $t \mapsto \exp(i \gamma(t))$ and
    that path is by definition of $P_{S^1_+}$ directed.
\end{example}

Any continuous map from a minimally directed space to a directed space is directed.
Similarly, any continuous map to a maximally directed space is directed.
By construction of the product of directed spaces
the continuous projection maps on both coordinates are directed:
a directed path in the product space is a pair of directed paths
and a projection returns the original directed path.
Similarly, if a continuous map $f : X \to Y$ is used to induce a direction on $X$
as in \cref{Induced Directed Spaces},
then $f$ becomes a directed map from $X$ to $Y$, where $X$ has the induced directedness.

In order to formalize the definition of a directed map in Lean,
we define the property \verb!directed_map.directed!,
which expresses exactly that a continuous map between two directed spaces
maps directed paths to directed paths:

\vspace{3mm}
\begin{leancode}
variables {α β : Type*} [directed_space α] [directed_space β]
def directed (f : C(α, β)) : Prop :=
  ∀ ⦃x y : α⦄ (γ : path x y),
        is_dipath γ → is_dipath (γ.map f.continuous_to_fun)
\end{leancode}
\vspace{3mm}

A directed map is then an extension of the \verb|continuous_map| structure
with a proof for being \verb|directed|.

\vspace{3mm}
\begin{leancode}
structure directed_map (α β : Type*) [directed_space α] [directed_space β]
  extends continuous_map α β :=
    (directed_to_fun : directed_map.directed to_continuous_map)
\end{leancode}
\vspace{3mm}

Within Lean, we use the notation $D(\alpha, \beta)$
for the type of directed maps between two spaces $\alpha$ and $\beta$.
Directed paths are also instances of directed maps,
because they map directed paths in~$I$ to monotone subparametrization of themselves.
\leanref{dipath.lean} contains definitions on how to convert the \verb|dipath| type
to the \verb|directed_map| type and the other way around.
These are called \verb|to_directed_map| and \verb|of_directed_map| respectively.

Directed spaces and directed maps form a category, which we will denote with \textbf{dTop}.
There are two functors $Min, Max : \textbf{Top} \to \textbf{dTop}$,
where $Min$ equips a topological space with the minimal directedness
and $Max$ equips a topological space with the maximal directedness.
If $U : \textbf{dTop} \to \textbf{Top}$ is the forgetful functor
that sends a directed space to its underlying topological space,
we obtain two adjunctions $Min \dashv U \dashv Max$ \cite{grandis2003directed}.

Within \textbf{dTop} we find an instance of pushouts as the following lemma shows.

\begin{lemma}\label{dTop pushout}
    Let $X \in \textbf{dTop}$ be a directed space
    and $X_1$ and $X_2$ two open subspaces such that $X = X_1 \cup X_2$.
    Take $X_0 = X_1 \cap X_2$ as the intersection of $X_1$ and $X_2$.
    Let $i_k : X_0 \to X_k$ and $j_k : X_k \to X$, with $k \in \{1, 2 \}$ be the inclusion maps.
    We then get a pushout square in \textbf{dTop}:
    \begin{center}
        \begin{tikzpicture}
        \node at (0,3) {$X_0$};
        \node at (3,3) {$X_1$};
        \node at (0,0) {$X_2$};
        \node at (3,0) {$X$};
        \draw[thick, ->] (0.3, 3) -- node[above] {$i_1$} (2.7, 3);
        \draw[thick, ->] (0, 2.7) -- node[left]  {$i_2$} (0, 0.3);
        \draw[thick, ->] (3, 2.7) -- node[right] {$j_1$} (3, 0.3);
        \draw[thick, ->] (0.3, 0) -- node[below] {$j_2$} (2.7, 0);
        \end{tikzpicture}
    \end{center}
\end{lemma}
\begin{proof}
    Let $Y \in \textbf{dTop}$ be another directed space and
    $f_1 : X_1 \to Y$ and $f_2 : X_2 \to Y$ two directed maps such that $f_1 \circ i_1 = f_2 \circ i_2$.
    We will now construct an unique directed map $f : X \to Y$ such that $f \circ j_1 = f_1$ and $f \circ j_2 = f_2$.
    Note that $f \circ j_k$ as a map is simply the restriction of $f$ to $X_k$.
    As $X$ is covered by $X_1$ and $X_2$, it follows that $f$ needs to be defined by
    \begin{gather*}
        f(x) = \begin{cases}
            f_1(x), & x \in X_1, \\
            f_2(x), & x \in X_2.
        \end{cases}
    \end{gather*}
    This already gives us uniqueness.
    If $U \subseteq Y$ is open,
    then $f^{-1}(U) = f_1^{-1}(U) \cup f_2^{-1}(U)$ is open as a union of open subsets
    and thus $f$ is continuous.
    In order to see that $f$ is directed,
    we need to use the Lebesgue Number Lemma \cite[p. 179-180]{munkres1975topology}.
    We will use it to cut up a path into pieces
    and then apply the two maps $f_1$ and $f_2$ independently and concatenate the results together.

    Let $\gamma \in P_X$ be any directed path.
    We have that $[0,1] = \gamma^{-1}(X) = \gamma^{-1}(X_1) \cup \gamma^{-1}(X_2)$,
    so $\gamma^{-1}(X_1)$ and $\gamma^{-1}(X_2)$ form an open cover of $[0,1]$.
    By applying the Lebesgue Number Lemma we can find an integer $n > 0$
    such that for all $1 \le i \le n$ we have that
    $[\tfrac{i}{n},\tfrac{i+1}{n}] \subseteq \gamma^{-1}(X_{k_i})$ with $k_i$ either $1$ or $2$,
    i.e. $\gamma$ is $n$-covered.
    With a suitable bijective and monotone reparametrization $\varphi : [0, 1] \to [0,1]$,
    we have that
    \begin{gather*}
        \gamma \circ \varphi =
        \gamma_{1, n} \odot (\gamma_{2, n} \odot \ldots (\gamma_{n-1, n} \odot \gamma_{n, n})).
    \end{gather*}
    Note that each of the paths $\gamma_{i, n}$ is directed
    as they are monotone subparametrizations of $\gamma$.
    We obtain:
    \begin{gather*}
        (f \circ \gamma) \circ \varphi =
        f \circ (\gamma \circ \varphi) =
        f \circ \left( \gamma_{1, n} \odot (\gamma_{2, n} \odot \ldots (\gamma_{n-1, n} \odot \gamma_{n, n})) \right) = \\
        (f_{k_1} \circ \gamma_{1, n}) \odot ((f_{k_2} \circ \gamma_{2, n}) \odot \ldots ((f_{k_{n-1}} \circ \gamma_{n-1, n}) \odot (f_{k_n} \circ \gamma_{n, n})))
    \end{gather*}
    The directedness of the maps $f_1$ and $f_2$ tells us
    that each of the paths $f_{k_i} \circ \gamma_{i, n}$ is directed.
    By the property of concatenation of directed paths we find
    that $(f \circ \gamma) \circ \varphi$ is directed.
    As $\varphi$ is a bijective monotone reparametrization, so is its inverse.
    This gives us that
        $f \circ \gamma = (f \circ \gamma) \circ \varphi \circ \varphi^{-1}$
    is directed, so $f$ is a directed map.
    From this, it follows that the above square is indeed a pushout square.
\end{proof}

From proof of the lemma, it follows that within the category \textbf{Top} a topological space $X$
covered by two open sets $X_1$ and $X_2$ also form a pushout square.
In the case that $X_1$ and $X_2$ are both closed,
they still form a pushout in \textbf{Top} but that is not guaranteed in \textbf{dTop}.
The latter is shown by the following counterexample.

\begin{example}\label{dTop pushout counterexample}
    Take $X = [0,1]$ (maximally directed),
    $X_1 = \{0\} \cup \left( \bigcup_{i = 0}^\infty \left[\tfrac{1}{2i+2}, \tfrac{1}{2i+1} \right] \right)$
    and
    $X_2 = \{0\} \cup \left( \bigcup_{i = 1}^\infty \left[\tfrac{1}{2i+1}, \tfrac{1}{2i} \right] \right)$.
    We have that
        $X_1 \cap X_2 = \{ 0 \} \cup \{ \tfrac{1}{n} \mid n \in \mathbb{Z}_{>1} \}$
    and
        $X_1 \cup X_2 = X$.
    Let $Y = [0, 1]$ with directed paths given by
        $P_Y = \{0_0\} \cup \{ \gamma : [0, 1] \to (0, 1] \mid \gamma \text{ continuous} \}$.
    The directed paths in $Y$ are thus paths contained in $(0, 1]$ and the constant path in $0$.
    Note that this collection does indeed satisfy the three properties of a directed space.

    The point $0$ in $X_1$ is not connected by any paths to any other points,
    so we claim that there is a directed map $f_1 : X_1 \to Y$ given by $f_1(x) = x$:
    we have that for any $\gamma \in P_{X_1}$ that $\gamma = 0_0$
    or $\gamma$ is contained in $\left[\tfrac{1}{2i+1}, \tfrac{1}{2i+2} \right]$ for some $i \geq 0$
    and thus in $(0, 1]$.
    In both cases, $f_1 \circ \gamma$ is directed in $Y$, making $f_1$ a directed map.
    Similarly we have a directed map $f_2 : X_2 \to Y$ given by $f_2(x) = x$.
    
    If $i_1 : X_1 \cap X_2 \to X_1$ and $i_2 : X_1 \cap X_2 \to X_2$ are the inclusion maps,
    then $(f_1 \circ i_1)(x) = x = (f_2 \circ i_2)(x)$ for all $x \in X_1 \cap X_2$,
    so $f_1 \circ i_1 = f_2 \circ i_2$ holds.
    If $X$ were the pushout of $X_1$ and $X_2$,
    there would have to be a unique directed map $f : X \to Y$
    such that $f\circ j_1 = f_1$ and $f \circ j_2 = f_2$.
    Clearly $f$ must be defined as in the proof of \cref{dTop pushout}:
    \begin{gather*}
        f(x) = \begin{cases}
            f_1(x), & x \in X_1, \\
            f_2(x), & x \in X_2,
        \end{cases} = \begin{cases}
            x, & x \in X_1, \\
            x, & x \in X_2,
        \end{cases} = x.
    \end{gather*}
    If $\gamma \in P_X$ is the directed path given by $\gamma(t) = t$,
    we would require that $f \circ \gamma \in P_Y$.
    As $f \circ \gamma$ is neither the constant path $0_0$ nor a path in $(0, 1]$,
    we find a contradiction with the directedness of~$f$.
    Therefore $X$ is not the pushout of $X_1$ and $X_2$, even though $X_1$ and $X_2$ are both closed.
\end{example}

From this example, it also follows that $Max$ does not preserve colimits,
as a pushout is an instance of a colimit.

\section{Directed Homotopies}\label{Ch Directed Homotopies}

In this section, we will look at directed homotopies and directed path homotopies.
These two concepts realize the idea of deformation,
while respecting the directedness of a directed space. 

\subsection{Homotopies} \label{Homotopies}

A directed homotopy is the deformation of one directed map into another.

\begin{definition}[Directed homotopy] Let $X$ and $Y$ be two directed spaces.
    A homotopy between two directed maps $f, g : X \to Y$ is a directed map $H : I \times X \to Y$
    such that for all $x \in X$ we have that $H(0, x) = f(x)$ and $H(1, x) = g(x)$,
    where the product $I \times X$ is taken between directed spaces, see \cref{Product Spaces}.
\end{definition}

Note that $I \times X$ has the product directedness.
We say that $H$ is a directed homotopy from $f$ to~$g$.
This order matters,
as unlike in the undirected case a directed homotopy cannot generally be reversed.
In our formalization,
we adhere to the method used in defining homotopies between continuous maps in MathLib, which can be found in
\href{https://github.com/leanprover-community/mathlib/blob/10bf4f8/src/topology/homotopy/basic.lean}{\texttt{topology/homotopy/basic.lean}}.
In an analogous manner,
the structure extends the \verb|directed_map (I × X) Y| structure and has two extra properties.

\vspace{3mm}
\begin{leancode}
structure dihomotopy (f₀ f₁ : directed_map X Y) extends D((I × X), Y) :=
  (map_zero_left' : ∀ x, to_fun (0, x) = f₀.to_fun x)
  (map_one_left' : ∀ x, to_fun (1, x) = f₁.to_fun x)
\end{leancode}
\vspace{3mm}

As a directed map is always a continuous map on the underlying topological spaces,
we can define how to convert a \verb|dihomotopy| to a \verb|homotopy|.
Conversely, if we are given a \verb|homotopy| and we know that it is directed,
we can obtain a \verb|dihomotopy|.

If $f : X \to Y$ is a directed map,
there is an identity homotopy $H$ from $f$ to $f$, given by $H(t, x) = f(x)$.
Also, if $G$ is a directed homotopy from $f$ to $g$ and $H$ a directed homotopy from $g$ to $h$,
we obtain a directed homotopy $G \otimes H$ from $f$ to $h$ given by
\begin{align*}
    (G \otimes H)(t, x) = \begin{cases}
        G(2t, x), & t \le \tfrac{1}{2}, \\
        H(2t - 1, x), & \tfrac{1}{2} < t.
    \end{cases}
\end{align*}

These constructions are called \verb|refl| and \verb|trans| in \leanref{directed\_homotopy.lean}.
In both cases the coercion of a \verb|homotopy| into a \verb|dihomotopy| is used,
by supplying the proofs that the obtained homotopies are directed.
Here we use the fact that MathLib already contains proofs that the constructed maps
are indeed homotopies, i.e. continuous and satisfying the two mapping properties.

\subsection{Path Homotopies}\label{Sch Path Homotopies}

\begin{definition}[Directed path homotopy]
    Let $X$ be a directed space and $x, y \in X$  two points.
    A directed path homotopy between two directed paths $\gamma_1, \gamma_2 \in P_X(x, y)$
    is a directed homotopy $H : I \times I \to X$ from $\gamma_1$ to $\gamma_2$
    such that additionally for all $t \in [0,1]$ we have that $H(t, 0) = x$ and $H(t, 1) = y$.
\end{definition}

In other words, a path homotopy is a homotopy between two paths that keeps both endpoints fixed.
Again we say that $H$ is a directed path homotopy from $\gamma_1$ to $\gamma_2$.
Between two paths $\gamma_1$ and~$\gamma_2$ in $I$ with the same endpoints
exists a path homotopy under the condition
that $\gamma_1(t) \le \gamma_2(t)$ for all $t \in I$ as the following example shows.

\begin{example}\label{Interpolate Path Homotopy}
    Let $t_0, t_1 \in I$ be two points and $\gamma_1, \gamma_2 \in P_I(t_0, t_1)$.
    If $\gamma_1(t) \le \gamma_2(t)$ for all $t \in I$,
    then there is a directed path homotopy $H$ from $\gamma_1$ to $\gamma_2$ given by
      $H(t, s) = (1 - t) \cdot \gamma_1(s) + t \cdot \gamma_2(s)$.
    It is continuous by continuity of paths, multiplication and addition.
    It can be shown that $H(a_0, b_0) \le H(a_1, b_1)$ if $a_0 \le a_1$ and $b_0 \le b_1$.
    From this, it follows that $H$ is directed,
    because a directed path in $I \times I$ is exactly a pair of monotone maps $I \to I$ by definition.
    
    Note that $H$ interpolates two paths $\gamma_1$ and $\gamma_2$.
    The formalized proof of it being a directed map can be found in the file \leanref{interpolate.lean}.
\end{example}

Let $x, y, z \in X$ be three points,
$\beta_1, \gamma_1 \in P_X(x, y)$ and $\beta_2, \gamma_2 \in P_X(y, z)$.
If there are two directed path homotopies
$G$ from $\beta_1$ to $\gamma_1$ and $H$ from $\beta_2$ to $\gamma_2$,
we can construct a directed path homotopy
$G \odot H$ from $\beta_1 \odot \beta_2$ to $\gamma_1 \odot \gamma_2$ given by
\begin{gather*}
    (G \odot H)(t, s) = \begin{cases}
        G(t, 2s), & s \le \tfrac{1}{2}, \\
        H(t, 2s - 1), & \tfrac{1}{2} < s.
    \end{cases}
\end{gather*}

Let $x, y \in X$ be two points and $\gamma_1, \gamma_2 \in P_X(x, y)$.
If there exists a path homotopy from $\gamma_1$ to~$\gamma_2$,
we will write $\gamma_1 \rightsquigarrow \gamma_2$.
This defines a relation on the set $P_X(x, y)$,
but that relation is not guaranteed to be an equivalence relation, as it is generally not symmetric.
This is due to the fact that the reversal of a directed path may not be directed.
The following lemma shows this.

\begin{lemma}
    Let $X$ be a directed space and $x, y \in X$.
    The relation $\rightsquigarrow$ on $P_X(x, y)$ is reflexive and transitive,
    but it is not always symmetric.
\end{lemma}
\begin{proof}
    In order to see that the relation is reflexive and transitive, we use the constructions from \cref{Homotopies}.
    If $\gamma \in P_X(x, y)$, then the directed homotopy given by $H(t, s) = \gamma(s)$
    satisfies the additional conditions of a path homotopy, so $\gamma \rightsquigarrow \gamma$.
    Similarly if $\gamma_1, \gamma_2, \gamma_3 \in P_X(x, y)$
    with $G$ a directed path homotopy from $\gamma_1$ to $\gamma_2$
    and $H$ a directed path homotopy from $\gamma_2$ to $\gamma_3$,
    then $G \otimes H$ is a directed homotopy from $\gamma_1$ to $\gamma_3$.
    It additionally holds that
    \begin{gather*}
        (G \otimes H)(t, 0) = \begin{cases}
            G(2t, 0), & t \le \tfrac{1}{2}, \\
            H(2t - 1, 0), & \tfrac{1}{2} < t.
        \end{cases} = 
        \begin{cases}
            x, & t \le \tfrac{1}{2}, \\
            x, & \tfrac{1}{2} < t
        \end{cases} = x,
    \end{gather*}
    and similarly that $(G \otimes H)(t, 1) = y$,
    so it is also a directed path homotopy and we find $\gamma_1 \rightsquigarrow \gamma_3$.
    A counterexample to symmetry is as follows:
    let $\gamma_1, \gamma_2 : I \to I$ be the paths given by
    \begin{gather*}
        \gamma_1(t) = \begin{cases}
            0, & t \le \tfrac{1}{2}, \\
            2t - 1, & \tfrac{1}{2} < t.
        \end{cases} \text{ \ \ and \ \ } \gamma_2(t) = \begin{cases}
            2t, & t \le \tfrac{1}{2}, \\
            1, & \tfrac{1}{2} < t.
        \end{cases}
    \end{gather*}
    We have that $\gamma_1 \rightsquigarrow \gamma_2$, because of \cref{Interpolate Path Homotopy}.
    On the other hand,
    if $H$ were a directed path homotopy from $\gamma_2$ to $\gamma_1$ we would require that
    $H(0, \tfrac{1}{2}) = \gamma_2(\tfrac{1}{2}) = 1$ and
    $H(1, \tfrac{1}{2}) = \gamma_1(\tfrac{1}{2}) = 0$, contradicting with directedness.
\end{proof}

In order get an equivalence relation on the set of directed paths between two points,
we will take the symmetric transitive closure of this relation.

\begin{definition}
    Let $X$ be a directed space and $x, y \in X$ two points.
    We say that two dipaths $\gamma_1, \gamma_2 \in P_X(x,y)$ are equivalent,
    or $\gamma_1 \simeq \gamma_2$,
    if there is an integer $n \geq 0$ together with dipaths $\beta_i \in P_X(x, y)$,
    for each $1 \le i \le n$, such that
    \begin{gather*}
      \gamma_1 \rightsquigarrow \beta_1 \leftsquigarrow \ldots \rightsquigarrow \beta_n \leftsquigarrow \gamma_2.
    \end{gather*}
\end{definition}
This alternating sequence of arrows is also called a zigzag.
As $\gamma_2 \leftsquigarrow \gamma_2$ holds for any path~$\gamma_2$ by reflexivity,
we can always assume that there is an odd number of paths in a zigzag
between two paths $\gamma_1$ and $\gamma_2$.
By taking $n = 0$,
it follows that $\gamma_1 \simeq \gamma_2$ holds if $\gamma_1 \rightsquigarrow \gamma_2$.
More precisely, $\simeq$ is the smallest equivalence relation on $P_X(x, y)$
such that that property holds \cite[p. 129]{leinster2016basic}.
As $\simeq$ is an equivalence relation,
we can talk about the equivalence classes of paths, denoted by~$[\gamma]$.
An important property of these equivalence classes is
that they are invariant under maps and path reparametrization.

\begin{lemma}\label{PathClassMap}
    Let $X, Y$ be directed spaces and $x, y \in X$.
    Let $\gamma_1, \gamma_2 \in P_X(x, y)$ and $f : X \to Y$ directed.
    If $\gamma_1 \simeq \gamma_2$, then $f \circ \gamma_1 \simeq f \circ \gamma_2$.
\end{lemma}
\begin{proof}
    Let $n > 0$ odd and $\beta_i \in P_X(x, y)$ for $1 \le i \le n$ such that
    \begin{gather*}
      \gamma_1 \rightsquigarrow \beta_1 \leftsquigarrow \beta_2 \rightsquigarrow \ldots \rightsquigarrow \beta_n \leftsquigarrow \gamma_2.
    \end{gather*}
    If $H : I \times I \to X$ is a directed path homotopy from $\gamma_1$ to $\beta_1$,
    then $f \circ H$ is a directed path homotopy from $f \circ \gamma_1$ to $f \circ \beta_1$.
    We find that $f \circ \gamma_1 \rightsquigarrow f \circ \beta_1$.
    Repeating this for all other arrows in the zigzag gives us
    \begin{gather*}
      f \circ \gamma_1 \rightsquigarrow f \circ \beta_1 \leftsquigarrow f \circ \beta_2 \rightsquigarrow \ldots \rightsquigarrow f \circ \beta_n \leftsquigarrow f \circ \gamma_2,  
    \end{gather*}
    We conclude that $f \circ \gamma_1 \simeq f \circ \gamma_2$.
\end{proof}

\begin{lemma}\label{PathClassReparam}
    Let $X$ be a directed space and $x, y \in X$.
    Let $\gamma \in P_X(x, y)$ and $\varphi, \varphi' : I \to I$ continuous and monotone
    with $\varphi(0) = \varphi'(0) = 0$ and $\varphi(1) = \varphi'(1) = 1$.
    Then $\gamma \circ \varphi \simeq \gamma \circ \varphi'$.
\end{lemma}
\begin{proof}
    As $\gamma$ is a directed map from $I$ to $X$,
    it is enough by \cref{PathClassMap} to show that $\varphi \simeq \varphi'$.
    Let $\beta_1 = \varphi \odot 0_1$ and $\beta_2 = 0_0 \odot \varphi'$.
    Then, by applying \cref{Interpolate Path Homotopy} three times, we obtain the zigzag 
    \begin{gather*}
        \varphi \rightsquigarrow \beta_1 \leftsquigarrow \beta_2 \rightsquigarrow \varphi'.
    \end{gather*}
    This shows that $\varphi \simeq \varphi'$, completing the proof.
\end{proof}

In the next section, we will construct the fundamental category of a directed space.
For that we need the following four additional equalities.

\begin{lemma}\label{PathClassProps}
    Let $X$ be a directed space and $x, y, z, w \in X$.
    Let $\beta_1, \gamma_1 \in P_X(x, y)$, $\beta_2, \gamma_2 \in P_X(y, z)$
    and $\gamma_3 \in P_X(z, w)$ such that $\beta_1 \simeq \gamma_1$ and $\beta_2 \simeq \gamma_2$.
    Then the following holds:
    \begin{enumerate}
        \item $\beta_1 \odot \beta_2 \simeq \gamma_1 \odot \gamma_2$ \label{PathClassProps1}
        \item $0_x \odot \gamma_1 \simeq \gamma_1$
        \item $\gamma_1 \odot 0_y \simeq \gamma_1$
        \item $(\gamma_1 \odot \gamma_2) \odot \gamma_3 \simeq \gamma_1 \odot (\gamma_2 \odot \gamma_3)$
    \end{enumerate}
\end{lemma}
\begin{proof}
    Statements 2, 3 and 4 are direct applications of \cref{PathClassReparam}
    as they are all reparametrizations.
    We will now show statement 1.
    Let $n, m > 0$ odd and $p_i, q_j \in P_X(x, y)$ for $1 \le i \le n$ and $1 \le j \le m$ such that
    \begin{gather*}
        \beta_1 \rightsquigarrow p_1 \leftsquigarrow p_2 \rightsquigarrow \ldots \rightsquigarrow p_n \leftsquigarrow \gamma_1, \\
        \beta_2 \rightsquigarrow q_1 \leftsquigarrow q_2 \rightsquigarrow \ldots \rightsquigarrow q_m \leftsquigarrow \gamma_2.
    \end{gather*}
    Let $G$ be a directed path homotopy from $\beta_1$ to $p_1$
    and $H$ be the identity homotopy from $\beta_2$ to $\beta_2$.
    Then $G \odot H$ is a directed path homotopy from $\beta_1 \odot \beta_2$ to $p_1 \odot \beta_2$.
    Repeating this, we obtain a zigzag
    \begin{gather*}
        \beta_1 \odot \beta_2 \rightsquigarrow p_1 \odot \beta_2  \leftsquigarrow p_2 \odot \beta_2 \rightsquigarrow \ldots \rightsquigarrow p_n \odot \beta_2  \leftsquigarrow \gamma_1 \odot \beta_2,
    \end{gather*}
    so $\beta_1 \odot \beta_2 \simeq \gamma_1 \odot \beta_2$.
    Analogously we obtain a zigzag
    \begin{gather*}
        \gamma_1 \odot \beta_2 \rightsquigarrow \gamma_1 \odot q_1 \leftsquigarrow \gamma_1 \odot q_2 \rightsquigarrow \ldots \rightsquigarrow \gamma_1 \odot q_m \leftsquigarrow \gamma_1 \odot \gamma_2.
    \end{gather*}
    This results in $\gamma_1 \odot \beta_2 \simeq \gamma_1 \odot \gamma_2$
    and combining both equivalences gives us $\beta_1 \odot \beta_2 \simeq \gamma_1 \odot \gamma_2$.
\end{proof}

The definition of a directed path homotopy and the three lemmas above have all been been formalized
in \leanref{directed\_path\_homotopy.lean}.
For the path homotopies, we followed the more general approach from MathLib,
where we first defined directed homotopies that satisfy some property $P$.
Thereafter we defined \verb|dihomotopy_rel| as directed homotopies
that are fixed on a select subset of points.
This is all defined in \leanref{directed\_homotopy.lean}.
A path homotopy is a homotopy that is fixed on both endpoints, that is, on $\{0, 1\} \subseteq I$,
so we can define a directed path homotopy as

\vspace{3mm}
\begin{leancode}
abbreviation dihomotopy (p₀ p₁ : dipath x₀ x₁) :=
    directed_map.dihomotopy_rel p₀.to_directed_map p₁.to_directed_map {0, 1}
\end{leancode}
\vspace{3mm}

As a directed homotopy is defined between two directed maps,
we need to convert both paths \verb|p₀| and \verb|p₁| to directed maps.
The construction $\odot$ is called \verb|hcomp| and $\otimes$ is called \verb|trans|.
If $f, g \in D(I, I)$ are two directed maps with $f(t) \le g(t)$ for all $t \in I$,
the definition \verb|dihomotopy.reparam| constructs a homotopy from $\gamma \circ f$ to $\gamma \circ g$.
This is done by composing $\gamma$ and the homotopy obtained from \cref{Interpolate Path Homotopy}.
If $H$ is a homotopy from $\gamma_1$ to $\gamma_2$ with $\gamma_1, \gamma_2 \in P_X(x, y)$,
and $f : X \to Y$ is a directed map, then the homotopy from $f \circ \gamma_1$ to $f \circ \gamma_2$
given by $f \circ H$ is exactly what \verb|dihomotopy.map| entails.

Now we can formalize the relations $\rightsquigarrow$ and $\simeq$.
These are called \verb|pre_dihomotopic| and \verb|dihomotopic| respectively.

\vspace{3mm}
\begin{leancode}
    def pre_dihomotopic : Prop := nonempty (dihomotopy p₀ p₁)
    def dihomotopic : Prop := eqv_gen pre_dihomotopic p₀ p₁
\end{leancode}
\vspace{3mm}

The term nonempty means exactly that there exists some dihomotopy, which corresponds with our definition of $\rightsquigarrow$.
\verb|eqv_gen| gives the smallest equivalence relation generated by a relation, which is exactly what we want.
The lemmas \verb|map|, \verb|reparam| and \verb|hcomp| in the namespace \verb|dihomotopic| now correspond with
\cref{PathClassMap}, \cref{PathClassReparam} and the first point of \cref{PathClassProps} respectively.

This gives us enough tools to construct the so called fundamental category.

\section{Fundamental Structures}\label{Ch Fundamental Structures}

In this section, we define two structures that contain information about directed paths up to
deformation in a directed space: the fundamental category and the fundamental monoid.

\subsection{The Fundamental Category}

Using the properties found in \cref{Sch Path Homotopies},
we can define a category that captures the information of all paths up to
directed deformation in a directed space.
This is the directed version of the fundamental groupoid.

\begin{definition}[Fundamental Category]
    Let $X$ be a directed space.
    The fundamental category of $X$, denoted by $\overrightarrow\Pi(X)$,
    is a category that consists of:
\begin{itemize}
    \item Objects: points $x \in X$.
    \item Morphisms: $\overrightarrow\Pi(X)(x, y) = P_X(x, y)/\simeq$.
    \item Composition: $[\gamma_2] \circ [\gamma_1] = [\gamma_1 \odot \gamma_2]$.
    \item Identity: $\text{id}_x = [0_x]$.
\end{itemize}  
\end{definition}

\begin{remark}
    The fact that this category is well defined follows from \cref{PathClassProps}.
    Due to property 1, composition is well defined.
    Due to properties 2 and 3,
    the constant path behaves as an identity and property 4 gives us associativity.
\end{remark}

Note that $\overrightarrow\Pi$ maps objects in \textbf{dTop} to objects in \textbf{Cat}.
It turns out that it can also be defined on morphisms making it into a functor.
\begin{definition}
    Let $f : X \to Y$ be a directed map.
    We define
      $\overrightarrow\Pi(f) : \overrightarrow\Pi(X) \to \overrightarrow\Pi(Y)$
    as the functor:
    \begin{itemize}
        \item On objects: $\overrightarrow\Pi(f)(x) = f(x)$.
        \item On morphisms: $\overrightarrow\Pi(f)([\gamma]) = [f \circ \gamma]$.
    \end{itemize}
\end{definition}
It is well behaved on morphisms, because of \cref{PathClassMap}.
It is straightforward to verify that $\overrightarrow\Pi(f)$ respects composition and identities.

In our formalization, we follow the construction of the fundamental groupoid in MathLib found in 
\href{https://github.com/leanprover-community/mathlib/blob/2a0ce625db/src/algebraic_topology/fundamental_groupoid/basic.lean}{\texttt{algebraic\_topology/fundamental\_groupoid/basic.lean}}
closely.
The implementation is found in the file \leanref{fundamental\_category.lean}.
The MathLib version has some auxiliary definitions for a reparametrization that show that the two paths
$(\gamma_1 \odot \gamma_2) \odot \gamma_3$ and $\gamma_1 \odot (\gamma_2 \odot \gamma_3)$
are equal with relation to $\simeq$ for compatible paths $\gamma_1, \gamma_2$ and $\gamma_3$.
In order to use these in our directed world, we need to show that this reparametrization is monotone.
This is enough to then define the fundamental category.
\vspace{3mm}
\begin{leancode}
    def fundamental_category (X : Type u) := X

    ...

    instance : category_theory.category (fundamental_category X) :=
    {
        hom := λ x y, dipath.dihomotopic.quotient x y,
        id := λ x, ⟦ dipath.refl x ⟧,
        comp := λ x y z, dipath.dihomotopic.quotient.comp,
        id_comp' := λ x y f, quotient.induction_on f
        (λ a, show ⟦ (dipath.refl x).trans a ⟧ = ⟦ a ⟧,
              from quotient.sound (eqv_gen.rel _ _ ⟨dipath.dihomotopy.refl_trans a⟩)),
        comp_id' := /- Proof omitted -/,
        assoc' := /- Proof omitted -/,
    }
\end{leancode}
\vspace{3mm}
The first definition makes sure that objects of the fundamental category are terms of type $X$.
We then show that \verb|fundamental_category| is an instance of a category by defining
the morphisms (\verb|hom|), identities (\verb|id|) and composition (\verb|comp|).
The morphisms between two objects \verb|x| and \verb|y| are given by \verb|dipath.dihomotopic.quotient x y|.
This is the quotient of \verb|dipath x y| under the \verb|dihomotopic| relation
and its definition can be found in \leanref{directed\_path\_homotopy.lean}.
The identity on \verb|x| is then the equivalence class (denoted by ⟦ ⟧) of the constant path in \verb|x|.
The composition of the equivalence classes of two compatible paths is defined as the equivalence
class of the concatenation of the two paths in \verb|dipath.dihomotopic.quotient.comp|.

Proofs that the fundamental category is indeed a category are given
by \verb|id_comp'|, \verb|comp_id'| and \verb|assoc'|.
The first one, \verb|id_comp'|, requires us to show that the directed paths
\verb|(dipath.refl x).trans a| and \verb|a| are dihomotopic.
For this, we use \verb|dipath.dihomotopy.refl_trans a|,
which is an explicit directed path homotopy from the path \verb|(dipath.refl x).trans a| to \verb|a|.
Its existence shows that the two paths are
\verb|pre_dihomotopic| and they are thus in the same equivalence class.

The file also contains the definition of the $\overrightarrow\Pi$-functor from \verb|dTop| to \verb|Cat|.
Analogous to the undirected MathLib implementation, we use the notation \verb|dπ| for this functor.

\subsection{The Fundamental Monoid}\label{SS Fundamental Monoid}

There is also a directed version of the fundamental group of a topological space at a point.
In the fundamental group every equivalence class has an inverse,
obtained by the equivalence class of a reversed path.
As directed paths are not reversible in general,
we do not always have inverses in the directed case.
This leads to the use of a monoid instead of a group.

\begin{definition}[Fundamental monoid]
    Let $X$ be a directed space and $x \in X$ a point.
    Then the fundamental monoid of $X$ at $x$, denoted by $\overrightarrow\pi(X, x)$, is given by
    the monoid $(\overrightarrow\Pi(X)(x, x), \circ, \text{id}_x)$:
    its elements are endomorphisms of $x$ in $\overrightarrow\Pi(X)$, the operation is composition
    and the neutral element is $\text{id}_x$.
\end{definition}

\begin{example}\label{Fundamental Monoid Unit Interval}
    Let $x \in I$ be a point.
    Let $\gamma \in P_I(x, x)$.
    Then $\gamma$ is monotone and $\gamma(0) = \gamma(1) = x$.
    It follows that $\gamma(t) = x$ for all $t \in I$, so $\gamma = 0_x$.
    From this, we can conclude that the only morphism in $\overrightarrow{\Pi}(I)(x, x)$ is the identity,
    so $\overrightarrow{\pi}(I, x)$ is the trivial monoid.
\end{example}

\begin{example}\label{Fundamental Monoid Unit Circle}
    We have that $\overrightarrow{\pi}(S^1_+, 1) \cong (\mathbb{N}, +, 0)$.
    In \cref{Van Kampen Application} we will support this claim by calculating a fundamental monoid in a finite version of the directed unit circle.
\end{example}

Whether we give the unit interval the minimal directedness, the rightward directedness or maximal directedness, the fundamental monoid at any point is the trivial monoid.
Their fundamental categories, on the other hand, are able to distinguish the differences in directedness.
There are respectively zero, one and two morphisms between two different objects.
We see that the fundamental monoid loses information that is contained in the fundamental category.

\section{The Van Kampen Theorem}\label{Ch Van Kampen Theorem}

In this section, we will state and prove the Van Kampen Theorem.
We follow the proof of Grandis and work out some of the details that were omitted.
In \cref{Van Kampen Formalization} we show how we have formalized this proof
by comparing the proof to the Lean code.
We conclude with an application of the Van Kampen Theorem in \cref{Van Kampen Application}. 

\subsection{The Van Kampen Theorem}\label{Van Kampen Theorem}

Before we state and prove the theorem,
we will define the notion of being covered for directed homotopies.

\begin{definition}
    Let $X$ be a directed space and $\mathcal{U}$ a cover of $X$.
    Let $H : I \times I \to X$ be a directed homotopy and $n, m > 0$ two integers.
    We say that $H$ is $(n, m)$-covered (by $\mathcal{U}$)
    if for all $1 \le i \le n$ and $1 \le j \le m$ the image of
      $\left[ \tfrac{i-1}{n}, \tfrac{i}{n} \right] \times \left[ \tfrac{j-1}{m}, \tfrac{j}{m} \right]$
    under $H$ is contained in some $U \in \mathcal{U}$.
\end{definition}

Once again, by the Lebesgue Number Lemma, for any homotopy $H$ and open cover $\mathcal{U}$ of $X$,
there are $n, m > 0$ such that $H$ is $(n, m)$-covered by $\mathcal{U}$.

\begin{theorem}[Van Kampen Theorem]\label{Van Kampen Theorem Theorem}
    Let $X$ be a directed space and $X_1$ and $X_2$ two open subspaces
    such that $X = X_1 \cup X_2$ and let $X_0 = X_1 \cap X_2$.
    Let $i_k : X_0 \to X_i$ and $j_k : X_k \to X$ be the inclusion maps, $k \in \{1, 2\}$.
    Then we obtain a pushout square in \textbf{Cat}:
    \begin{center}
        \begin{tikzpicture}
        \node at (0,3) {$\vec\Pi(X_0)$};
        \node at (3,3) {$\vec\Pi(X_1)$};
        \node at (0,0) {$\vec\Pi(X_2)$};
        \node at (3,0) {$\vec\Pi(X)$};
        \draw[thick, ->] (0.7, 3) -- node[above=1mm] {$\vec\Pi(i_1)$} (2.45, 3);
        \draw[thick, ->] (0, 2.6) -- node[left=1mm]  {$\vec\Pi(i_2)$} (0, 0.4);
        \draw[thick, ->] (3, 2.6) -- node[right=1mm] {$\vec\Pi(j_1)$} (3, 0.4);
        \draw[thick, ->] (0.7, 0) -- node[below=1mm] {$\vec\Pi(j_2)$} (2.45, 0);
        \end{tikzpicture}
    \end{center}
\end{theorem}
\begin{proof}
    As $j_1 \circ i_1 = j_2 \circ i_2$ and $\vec\Pi$ is a functor, the square is commutative.
    It remains to show it satisfies the property of a pushout square.
    Let $\mathcal{C}$ be any category and $F_1 : \vec\Pi(X_1) \to \mathcal{C}$ and
    $F_2 : \vec\Pi(X_2) \to \mathcal{C}$ be two functors such that
        $F_1 \circ \vec\Pi(i_1) = F_2 \circ \vec\Pi(i_2)$. 
    We will explicitly construct a functor $F : \vec\Pi(X) \to \mathcal{C}$
    such that $F \circ \vec\Pi(j_1) = F_1$ and $F \circ \vec\Pi(j_2) = F_2$.
    The construction will show that this functor is necessarily unique with this property.

    \textbf{Step 1)}
    The objects of $\vec\Pi(X)$ are exactly the points of $X$.
    If an object $x \in \vec\Pi(X)$ is also contained in $\vec\Pi(X_1)$,
    it holds that $F(x) = F (j_1(x)) = (F \circ \vec\Pi(j_1))(x)$.
    The desired condition $F \circ \vec\Pi(j_1) = F_1$ then requires us to define $F(x) = F_1(x)$.
    A similar argument gives us that if $x \in \vec\Pi(X_2)$ then $F(x) = F_2(x)$.
    As $X_1$ and $X_2$ cover $X$, we have that for all $x \in \vec\Pi(X)$ that
    \begin{gather*}
        F(x) = \begin{cases}
            F_1(x), & x \in X_1, \\
            F_2(x), & x \in X_2.
        \end{cases}
    \end{gather*}
    By the property that $F_1 \circ \vec\Pi(i_1) = F_2 \circ \vec\Pi(i_2)$ this is well defined,
    so we know how $F$ must behave on objects.

    \textbf{Step 2)}
    Let $[\gamma] : x \to y$ be a morphism in $\vec\Pi(X)$.
    Then there is a $n > 0$ such that $\gamma$ is $n$-covered by the open cover $\{X_1, X_2\}$,
    with $\gamma_{i, n}$ covered by $X_{k_i}$, $k_i \in \{ 1, 2\}$.
    One important thing to note is that $\gamma_{i, n}$ can be both seen as a path in $X$ and
    as a path in $X_{k_i}$ by restricting its codomain.
    This matters when we talk about $[\gamma_{i, n}]$, as it could be a morphism in $\vec{\Pi}(X)$
    and in $\vec{\Pi}(X_{k_i})$.
    Within this proof will always consider it as a morphism in $\vec{\Pi}(X_{k_i})$ and
    use $[j_{k_i} \circ \gamma_{i, n}]$ for the morphism in $\vec{\Pi}(X)$.
    Note that we have that
      $[\gamma] = [j_{k_n} \circ \gamma_{n, n}] \circ \ldots  \circ [j_{k_1} \circ  \gamma_{1, n}]$
    in $\vec{\Pi}(X)$, as $\gamma$ is equal to
      $\gamma_{1, n} \odot (\gamma_{2, n} \odot \ldots (\gamma_{n-1, n} \odot \gamma_{n, n}))$
    up to reparametrization.
    Because we want $F$ to be a functor and thus to respect composition, we find that necessarily
    \begin{align*}
        F [\gamma] &= F ([j_{k_n} \circ \gamma_{n, n}] \circ \ldots  \circ [j_{k_1} \circ \gamma_{1, n}]) \\
                    &= F [j_{k_n} \circ \gamma_{n, n}] \circ \ldots  \circ F[j_{k_1} \circ \gamma_{1, n}] \\
                    &= F \left(\vec\Pi(j_{k_n}) [\gamma_{n, n}]\right) \circ \ldots  \circ F \left(\vec\Pi(j_{k_1})[\gamma_{1, n}]\right) \\
                    &= (F  \circ \vec\Pi(j_{k_n})) [\gamma_{n, n}] \circ \ldots  \circ (F  \circ \vec\Pi(j_{k_1})) [\gamma_{1, n}] \\
                    &= F_{k_n} [\gamma_{n, n}] \circ \ldots  \circ F_{k_1}[\gamma_{1, n}].
    \end{align*}
    As multiple choices were made, we need to make sure that $F$ is well defined this way.
    We do this by defining a map $F' : P_X \to \mathrm{Mor}(\mathcal{C})$,
    where $\mathrm{Mor}$ is the
    collection of all morphisms in $\mathcal{C}$.
    The map is given by
    \begin{gather*}
        F'(\gamma) = F_{k_n} [\gamma_{n, n}] \circ \ldots  \circ F_{k_1}[\gamma_{1, n}],
    \end{gather*}
    where $\gamma$ is $n$-covered with $\gamma_{i, n}$ covered by $X_{k_i}$.
    Firstly, we will show that this map is well defined.
    Secondly, we show that $F'$ respects equivalence classes.
    From this it follows that $F$ is well defined,
    as it is simply $F'$ descended on equivalence classes.

    \textbf{Step 3)}
    We first need to make sure that $F'$ does not depend on any choices of $k_i$.
    In the case that $\gamma_{i, n}$ is covered by both $X_1$ and $X_2$,
    the value of $k_i$ can be either 1 or 2.
    The condition that  $F_1 \circ \vec\Pi(i_1) = F_2 \circ \vec\Pi(i_2)$
    assures us that both options give us the same morphism.

    \textbf{Step 4)}
    The second choice we made is that of $n$.
    It is possible that $\gamma$ is also $m$-covered for another integer $m > 0$,
    with $\gamma_{j, m}$ being contained in $X_{p_j}$.
    We want to show that
    \begin{gather*}
        F_{k_n} [\gamma_{n, n}] \circ \ldots  \circ F_{k_1}[\gamma_{1, n}] =
        F_{p_m} [\gamma_{m, m}] \circ \ldots  \circ F_{p_1}[\gamma_{1, m}].
    \end{gather*}
    If we refine the partition of $\gamma$ in $n$ pieces into a partition of $mn$ pieces,
    that partition will surely also be partwise covered.
    Let $l_i \in \{1, 2\}$ for all $1 \le i \le mn$
    such that $\gamma_{i, mn}$ is covered by $X_{l_i}$.
    We now claim that for all $1 \le i \le n$ it holds that
      $F_{k_i}[\gamma_{i, n}] =
       F_{l_{mi}}[\gamma_{mi, mn}] \circ \ldots \circ F_{l_{m(i-1) + 1}}[\gamma_{m(i-1)+1, mn}]$.
    As $\gamma_{m(i-1) + k, mn}$ with $1 \le k \le m$ is a
    part of $\gamma_{i, n}$, we may assume that $l_{m(i-1) + k} = k_i$.
    This is because $F_1$ and $F_2$ agree on $X_1 \cap X_2$.
    As $F_{k_i}$ is a functor, the claim now follows because functors respect composition and because
    $\gamma_{i, n}$ is exactly the concatenation of all the smaller paths up to reparametrization.
    By a similar claim for $F_{p_j}[\gamma_{j, m}]$ we find:
    \begin{align*}
        F_{k_n} [\gamma_{n, n}] \circ \ldots  \circ F_{k_1}[\gamma_{1, n}]
            &= F_{l_{mn}}[\gamma_{mn, mn}] \circ \ldots \circ F_{l_1}[\gamma_{1, mn}] \\
            &= F_{p_m} [\gamma_{m, m}] \circ \ldots  \circ F_{p_1}[\gamma_{1, m}].
    \end{align*}
    We conclude that the definition is independent of the value of $n$.
    This makes $F'$ well defined.

    \textbf{Step 5)}
    Before we verify that $F'$ is independent of the choice of representative $\gamma$,
    we will first show that $F'$ satisfies two properties:
    \begin{gather}
        \forall x \in \vec\Pi(X) : F'(0_x) = \text{id}_{F(x)}. \label{Fprime property 1} \\
        \forall \gamma \in P_X(x, y), \delta \in P_X(y, z) : F'(\gamma \odot \delta) = F'(\delta) \circ F'(\gamma). \label{Fprime property 2}
    \end{gather}

    Let $x \in \vec\Pi(X)$ be given.
    If $x \in X_1$, then $0_x$ is already contained in $X_1$ and so by definition of $F'$
    we find $F' (0_x) = F_1 [0_x] = \text{id}_{F_1(x)} = \text{id}_{F(x)}.$
    Otherwise it holds that
      $x \in X_2$, so $F'(0_x) = F_2 [0_x] = \text{id}_{F_2(x)} = \text{id}_{F(x)}$.
    This proves \cref{Fprime property 1}.

    Let $\gamma \in P_X(x, y), \delta \in P_X(y, z)$ be two paths in $\vec\Pi(X)$.
    We can then find a $n$ such that both $\gamma$ and $\delta$ are $n$-covered,
    with $\gamma_{i, n}$ contained in $X_{k_i}$ and $\delta_{i, n}$ contained in $X_{p_i}$.
    It is then true that $\gamma \odot \delta$ is $2n$-covered as it holds that
    \begin{gather*}
        (\gamma \odot \delta)_{i, 2n} = \begin{cases}
            \gamma_{i, n}, & i \le n, \\
            \delta_{i - n, n}, & i > n. \\
        \end{cases} 
    \end{gather*}
    We find:
    \begin{gather*}
        F' (\delta \odot \gamma) = \\
        F_{p_{n}} [(\delta \odot \gamma)_{2n, 2n}] \circ \ldots \circ F_{p_1} [(\delta \odot \gamma)_{n+1, 2n}] \circ 
        F_{k_n} [(\delta \odot \gamma)_{n, 2n}] \circ \ldots \circ F_{k_1} [(\delta \odot \gamma)_{1, 2n}] = \\
        (F_{p_n} [\delta_{n, n}] \circ \ldots \circ F_{p_1} [\delta_{1, n}]) \circ 
        (F_{k_n} [\gamma_{n, n}] \circ \ldots \circ F_{k_1} [\gamma_{1, n}]) = F' (\delta) \circ F' (\gamma).
    \end{gather*}
    This shows that \cref{Fprime property 2} holds.

    \textbf{Step 6)}
    We will now show that $F'$ respects equivalence classes.
    Then it can descend to the quotient and it follows that $F$ is well defined.
    If $[\gamma] = [\delta]$ in $\vec{\Pi}(X)$ with $\delta$ another path from $x$ to $y$,
    we want that
    \begin{gather}\label{VKT ClassInpedence}
        F'(\gamma) = F'(\delta).
    \end{gather}
    Because of the way the equivalence class is defined,
    it is enough to show this for $\gamma$ and $\delta$ such that $\gamma \rightsquigarrow \delta$.
    Let in that case a directed path homotopy $H$ from $\gamma$ to $\delta$ be given.
    We take $n, m > 0$ such that $H$ is $(n, m)$-covered by $\{X_1, X_2\}$.
    Firstly assume that $n > 1$.
    Restricting $H$ to the rectangle
      $\left[0, \tfrac{1}{n}\right] \times \left[0, 1\right]$
    gives us a directed path homotopy $H_1$ from $\gamma$ to the directed path $\eta$ given by
      $\eta(t) = H\left(\tfrac{1}{n}, t\right)$.
    By restricting $H$ to the rectangle
      $\left[\tfrac{1}{n}, 1 \right] \times \left[0, 1\right]$
    we get a directed path homotopy $H_2$ from $\eta$ to $\delta$.
    It is clear that $H_1$ is $(1, m)$-covered and that $H_2$ is $(n-1, m)$-covered.
    By applying induction on $n$, we can conclude that it is enough to show that \cref{VKT ClassInpedence}
    holds for $(1, m)$-covered directed path homotopies,
    as we would obtain that $F'(\gamma) = F'(\eta) = F'(\delta)$.
    
    \textbf{Step 7)}
    We will prove the case where $H$ is $(1, m)$-covered by showing a more general statement:

    Let $H$ be any directed homotopy -- not necessarily a path homotopy -- from one path
    $\gamma \in P_X(x, y)$ to another path $\delta \in P_X(x', y')$ that is $(1, m)$-covered, $m > 0$.
    Let $\eta_0$ be the path given by $\eta_0(t) = H(t, 0)$
    and $\eta_1$ be given by $\eta_1(t) = H(t, 1)$.
    Then $F'(\eta_0 \odot \delta) = F'(\gamma \odot \eta_1)$.
    We do this by induction on $m$.

    In the case that $m = 1$, we have a homotopy contained in $X_1$ or $X_2$.
    Without loss of generality, we can assume it is contained in $X_1$.
    Let $\Gamma_1$ be the directed homotopy given by
    $\Gamma_1(t, s) = \eta_0(\text{min}(t, s))$ from $0_x$ to $\eta_0$.
    Let $\Gamma_2$ be the directed homotopy given by
    $\Gamma_2(t, s) = \eta_1(\text{max}(t, s))$ from $\eta_1$ to $0_{y'}$.
    We then can construct a directed path homotopy from $(0_x \odot \gamma) \odot \eta_1$
    to $(\eta_0 \odot \delta) \odot 0_{y'}$ given by $(\Gamma_1 \odot H) \odot \Gamma_2$:
    
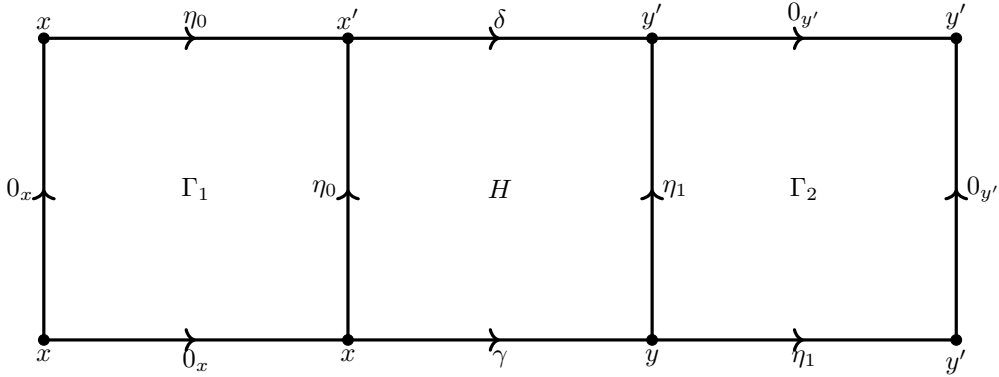
\begin{center}
        \begin{tikzpicture}
        \begin{scope}[very thick,decoration={
            markings,
            mark=at position 0.5 with {\arrow{>}}}
            ] 
            \draw[postaction={decorate}] (-6,0) -- node[below] {$0_x$} (-2,0);
            \draw[postaction={decorate}] (-2,0) -- node[below] {$\gamma$} ( 2,0);
            \draw[postaction={decorate}] ( 2,0) -- node[below] {$\eta_1$} ( 6,0);

            \draw[postaction={decorate}] (-6,4) -- node[above] {$\eta_0$} (-2,4);
            \draw[postaction={decorate}] (-2,4) -- node[above] {$\delta$} ( 2,4);
            \draw[postaction={decorate}] ( 2,4) -- node[above] {$0_{y'}$} ( 6,4);

            \draw[postaction={decorate}] (-6,0) -- node[left] {$0_{x}$} (-6,4);
            \draw[postaction={decorate}] (-2,0) -- node[left] {$\eta_0$} (-2,4);
            \draw[postaction={decorate}] ( 2,0) -- node[right] {$\eta_1$} ( 2,4);
            \draw[postaction={decorate}] ( 6,0) -- node[right] {$0_{y'}$} ( 6,4);
        \end{scope}
        \filldraw[black] (-6,0) circle (2pt) node[anchor=north]{$x$};
        \filldraw[black] (-2,0) circle (2pt) node[anchor=north]{$x$};
        \filldraw[black] (2,0) circle (2pt) node[anchor=north]{$y$};
        \filldraw[black] (6,0) circle (2pt) node[anchor=north]{$y'$};

        \filldraw[black] (-6,4) circle (2pt) node[anchor=south]{$x$};
        \filldraw[black] (-2,4) circle (2pt) node[anchor=south]{$x'$};
        \filldraw[black] (2,4) circle (2pt) node[anchor=south]{$y'$};
        \filldraw[black] (6,4) circle (2pt) node[anchor=south]{$y'$};

        \node at (-4, 2) {$\Gamma_1$};
        \node at (0, 2) {$H$};
        \node at (4, 2) {$\Gamma_2$};
        \end{tikzpicture}
    \end{center}
    It is a directed path homotopy because
    $\Gamma_1(t, 0) = \eta_0(\text{min}(t, 0)) = \eta_0(0) = x$ and
    $\Gamma_2(t, 1) = \eta_1(\text{max}(t, 1)) = \eta_1(1) = y'$ for all $t \in I$.
    As $\eta_0, \eta_1$ and $H$ are all contained in $X_1$,
    this directed path homotopy will be contained in $X_1$ as well.
    We find that $[\gamma \odot \eta_1] = [\eta_0 \odot \delta]$ in $\vec\Pi(X_1)$.
    This gives us that $F'(\gamma \odot \eta_1) =
      F_1 [\gamma \odot \eta_1] = F_1 [\eta_0 \odot \delta] = F'(\eta_0 \odot \delta)$.
    
    Let now $m > 1$ and assume the statement holds for $(1, m-1)$-covered homotopies.
    We can restrict $H$ to
      $[0,1] \times \left[0, \tfrac{m-1}{m}\right]$
    to obtain a $(1, m-1)$-covered homotopy $H_1$ and we can restrict $H$ to
      $[0,1] \times \left[\tfrac{m-1}{m}, 1\right]$
    to obtain a $(1, 1)$-covered homotopy $H_2$:
    
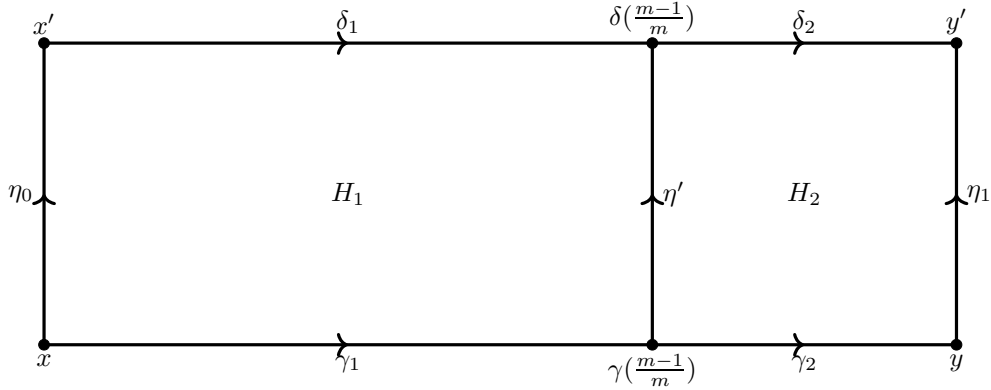
\begin{center}
        \begin{tikzpicture}
        \begin{scope}[very thick,decoration={
            markings,
            mark=at position 0.5 with {\arrow{>}}}
            ] 
            \draw[postaction={decorate}] (-6,0) -- node[below] {$\gamma_1$} (2,0);
            \draw[postaction={decorate}] ( 2,0) -- node[below] {$\gamma_2$} ( 6,0);

            \draw[postaction={decorate}] (-6,4) -- node[above] {$\delta_1$} (2,4);
            \draw[postaction={decorate}] ( 2,4) -- node[above] {$\delta_2$} ( 6,4);

            \draw[postaction={decorate}] (-6,0) -- node[left] {$\eta_0$} (-6,4);
            \draw[postaction={decorate}] ( 2,0) -- node[right] {$\eta'$} ( 2,4);
            \draw[postaction={decorate}] ( 6,0) -- node[right] {$\eta_1$} ( 6,4);
        \end{scope}
        \filldraw[black] (-6,0) circle (2pt) node[anchor=north]{$x$};
        \filldraw[black] (2,0) circle (2pt) node[anchor=north]{$\gamma(\tfrac{m-1}{m})$};
        \filldraw[black] (6,0) circle (2pt) node[anchor=north]{$y$};
        
        \filldraw[black] (-6,4) circle (2pt) node[anchor=south]{$x'$};
        \filldraw[black] (2,4) circle (2pt) node[anchor=south]{$\delta(\tfrac{m-1}{m})$};
        \filldraw[black] (6,4) circle (2pt) node[anchor=south]{$y'$};

        \node at (-2, 2) {$H_1$};
        \node at (4, 2) {$H_2$};
        \end{tikzpicture}
    \end{center}
    Note that $F'(\gamma) = F'(\gamma_2) \circ F'(\gamma_1)$ by definition,
    because $\gamma_1$ is $(m-1)$-covered, $\gamma_2$ is $1$-covered and $\gamma$ is $m$-covered.
    Similarly it holds that $F'(\delta) = F'(\delta_2) \circ F'(\delta_1)$.
    We find:
    \begin{align*}
        F'(\gamma \odot \eta_1) &= F'(\eta_1) \circ F'(\gamma) \\
        &= F'(\eta_1) \circ (F'(\gamma_2) \circ F'(\gamma_1)) \\ 
        &= (F'(\eta_1) \circ F'(\gamma_2)) \circ F'(\gamma_1) \\
        &= (F'(\delta_2) \circ F'(\eta')) \circ F'(\gamma_1) &&\qquad\text{(Case } m = 1 \text{)} \\
        &= F'(\delta_2) \circ (F'(\eta') \circ F'(\gamma_1)) \\ 
        &= F'(\delta_2) \circ (F'(\delta_1) \circ F'(\eta_0)) &&\qquad\text{(Induction Hypothesis)} \\
        &= (F'(\delta_2) \circ F'(\delta_1)) \circ F'(\eta_0) \\
        &= F'(\delta) \circ F'(\eta_0) \\
        &= F'(\eta_0 \odot \delta).
    \end{align*}
    
    This proves the statement. From the statement we find that \cref{VKT ClassInpedence} holds: 
    \begin{gather*}
        F'(\delta) = F'(\delta) \circ \text{id}_x = F'(\delta) \circ F'(0_x) = F'(0_x \odot \delta) = \\
        F'(\gamma \odot 0_y) = F'(0_y) \circ F'(\gamma) = \text{id}_x \circ F'(\gamma) = F'(\gamma).
    \end{gather*}
    Here, the fourth equality follows from the statement.
    We conclude that $F$ is well defined.

    \textbf{Step 8)}
    As we have that $F[\gamma] = F'(\gamma)$,
    it is immediate that $F$ is a functor by \cref{Fprime property 1} and \cref{Fprime property 2}.
    The equalities $F \circ \vec\Pi(j_1) = F_1$ and
      $F \circ \vec\Pi(j_2) = F_2$ are by construction true:
    if $\gamma$ is covered by $X_1$, then $\gamma_{1, 1}$ is as well,
    so $(F \circ \vec\Pi(j_1))[\gamma] = F[\gamma] = F'(\gamma) = F_1[\gamma_{1,1}] = F_1[\gamma].$
    Here the second $[\gamma]$ is a morphism in $\vec{\Pi}(X)$ and the others are in $\vec{\Pi}(X_1)$.
    We conclude that the commutative square is indeed a pushout.
\end{proof}

We see that a space $X$ being covered by two open subspaces $X_1$ and $X_2$
is a sufficient condition for the conclusion of this theorem to hold.
The essence of the proof is that any directed path homotopy can be covered by a
grid of rectangles such that it maps each rectangle into either $X_1$ or $X_2$.
This holds as a consequence of the Lebesgue Number Lemma if $X_1$ and $X_2$ are open. 
Without much effort, that covering property can be shown to be also true
if $X = X_1^\circ \cup X_2^\circ$,
where $S^\circ$ is the topological interior of a subset $S \subseteq X$.
Therefore, the condition of the Van Kampen Theorem can be relaxed to the case where $X_1$ and $X_2$
are not necessarily open and $X = X_1^\circ \cup X_2^\circ$.
That is how the theorem is stated in the original work of Grandis \cite[p. 306]{grandis2003directed}.

It is however not necessary that \textit{every} directed path homotopy can be covered.
The fundamental category of $I$ is a pushout of the fundamental categories
  $[0, \tfrac{1}{2}]$ and $[\tfrac{1}{2}, 1]$,
but there are directed path homotopies that are not covered by rectangles.
Take, for example, the interpolation homotopy (see \cref{Interpolate Path Homotopy})
between $0_0 \odot \gamma$ and $\gamma \odot 0_1$, with $\gamma$ the identity map on $I$.
The key is that it is still possible to find \textit{some}
homotopy between these two paths that is $(n, m)$-covered.

The original Van Kampen Theorem, stated by Egbert van Kampen,
was concerned with the fundamental group of a space \cite{van1933connection}.
The fundamental group is the undirected version of the fundamental monoid.
The version of the theorem for fundamental groups only requires the additional condition
that each of $X$, $X_1$, $X_2$ and $X_1 \cap X_2$ is path connected.
The fundamental monoid of a directed space in a point, however,
is not guaranteed to form a pushout under the same conditions as the following example shows.

\begin{example}\label{Counterexample Fundamental Monoid}
    We take $X$ to be the directed unit circle,
    together with the horizontal diameter directed leftward. 
    Take $X_1$ to be the top semicircle together with the diameter
    and $X_2$ to be the bottom semicircle together with the diameter.
    Expand them both a little bit to make sure that they are open subspaces.
    \begin{center}
        \begin{tikzpicture}[scale=0.6, every node/.style={scale=0.6}]
        \begin{scope}[very thick,decoration={
            markings,
            mark=at position 0.5 with {\arrow{>}}}
            ] 
            \draw[postaction={decorate}] (2,0)--(-2,0);
            \draw[postaction={decorate}] (2, 0) arc (0:180:2 and 2);
            \draw[postaction={decorate}] (-2, 0) arc (180:360:2 and 2);
        \end{scope}
        \filldraw[black] (-2,0) circle (2pt) node[anchor=east]{$x$};
        \end{tikzpicture}
        \qquad
        \begin{tikzpicture}[scale=0.6, every node/.style={scale=0.6}]
        \useasboundingbox (-2.3,-2.3) rectangle (2.3,2.3);
        \begin{scope}[very thick,decoration={
            markings,
            mark=at position 0.5 with {\arrow{>}}}
            ] 
            \draw[postaction={decorate}] (2,0)--(-2,0);
            \draw[postaction={decorate}] (1.879, -0.684) arc (-20:200:2 and 2);
        \end{scope}
        \filldraw[black] (-2,0) circle (2pt) node[anchor=east]{$x$};
        \draw[black, very thick] (-1.854, -0.75) circle (2pt);
        \draw[black, very thick] (1.854, -0.75) circle (2pt);
        \end{tikzpicture}
        \qquad
        \begin{tikzpicture}[scale=0.6, every node/.style={scale=0.6}]
            \begin{scope}[very thick,decoration={
                markings,
                mark=at position 0.5 with {\arrow{>}}}
                ] 
                \draw[postaction={decorate}] (2,0)--(-2,0);
                \draw[postaction={decorate}] (-1.879, 0.684) arc (160:380:2 and 2);
            \end{scope}
            \filldraw[black] (-2,0) circle (2pt) node[anchor=east]{$x$};
            \draw[black, very thick] (-1.854, 0.75) circle (2pt);
            \draw[black, very thick] (1.854, 0.75) circle (2pt);
        \end{tikzpicture} \\
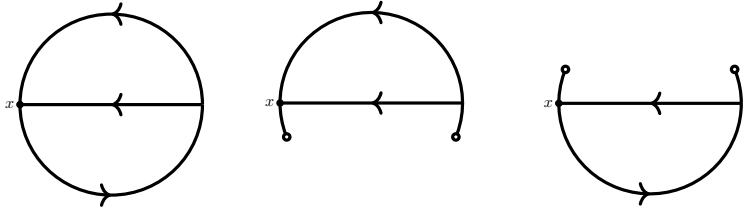
    The three spaces $X$, $X_1$ and $X_2$, from left to right.
    \end{center}
    In $X_1$, and therefore also in $X_1 \cap X_2$,
    the only path from $x$ to $x$ is the constant path,
    so it follows that $\vec{\pi}(X_1, x) \cong 0$ and $\vec{\pi}(X_1 \cap X_2, x) \cong 0$.
    Endomorphisms of $x$ in $X_2$ behave like endomorphisms of $1$ in $S^1_+$,
    so $\vec{\pi}(X_2, x) \cong (\mathbb{N}, +, 0)$.
    If $\vec{\pi}(X, x)$ were the pushout,
    we would find that $\vec{\pi}(X, x) \cong (\mathbb{N}, +, 0)$.
    This is false, as $\vec{\pi}(X, x) \cong (\mathbb{N}, +, 0) * (\mathbb{N}, +, 0)$
    --- the free product of $\mathbb{N}$ with itself.
    We will support this claim in \cref{Van Kampen Application}.
\end{example}

Under the right conditions, a Van Kampen type theorem for the fundamental monoid holds.
One such condition can be extracted from the work of Bubenik \cite{bubenik2009Models}.
Let $X_1$ and $X_2$ be two open subspaces covering the directed space $X$ with base point $x$,
both containing $x$.
It is now sufficient that for any endomorphism $[\gamma] : x \to x$
we can split it as $[\gamma] = [\gamma_n] \circ \ldots \circ [\gamma_1]$ with
$[\gamma_i] : x \to x$ and $\gamma_i$ contained in either $X_1$ or $X_2$.

In \cref{Counterexample Fundamental Monoid} this condition is not fulfilled
as one counterclockwise loop along the circle is first contained in $X_2$ and then in $X_1$.
There is no way to factor that loop into endomorphisms of~$x$ in such a way
that each morphism is contained in $\vec\Pi(X_1)$ or $\vec\Pi(X_2)$.

\subsection{Formalization}\label{Van Kampen Formalization}
In the formalization of \cref{Van Kampen Theorem Theorem}
we follow the constructive nature of its proof.
It can be found in \leanref{directed\_van\_kampen.lean}.
We have the following global variables,
corresponding with the assumptions of the Van Kampen Theorem: 
\vspace{3mm}
\begin{leancode}
    variables {X : dTop.{u}} {X₁ X₂ : set X}
    variables (hX : X₁ ∪ X₂ = set.univ)
    variables (X₁_open : is_open X₁) (X₂_open : is_open X₂)
\end{leancode}
\vspace{3mm}

Like in the proof,
we introduce a category $C$ and two functors $F_1 : \vec\Pi(X_1) \to C$ and $F_2 : \vec\Pi(X_2) \to C$.
Using these we are going to explicitly construct a functor
from $\vec\Pi(X)$ to $C$ and show that it is unique.
We will use that to prove that we indeed have a pushout square.

\vspace{3mm}
\begin{leancode}
    variables {C : category_theory.Cat.{u u}}
        (F₁ : (dπₓ (dTop.of X₁) ⟶ C))
        (F₂ : (dπₓ (dTop.of X₂) ⟶ C))
        (h_comm : (dπₘ i₁) ⋙ F₁ = ((dπₘ i₂) ⋙ F₂))
\end{leancode}
\vspace{3mm}

Here \verb|i₁| and \verb|i₂| are the inclusion maps as in the statement of the Van Kampen Theorem.
They are obtained by \verb|dTop.directed_subset_hom|, defined in
\leanref{dTop.lean}.
This defines the inclusion morphism $X_0 \to X_1$ in \textbf{dTop}
in the case that $X_0 \subseteq X_1 \subseteq X$.
In this case we have $X_0 = X_1 \cap X_2$. We first define the functor on objects (\textbf{Step 1}).

\vspace{3mm}
\begin{leancode}
    def functor_obj (x : dπₓ X) : C :=
        or.by_cases ((set.mem_union x X₁ X₂).mp (filter.mem_top.mpr hX x))
            (λ hx, F₁.obj ⟨x, hx⟩) (λ hx, F₂.obj ⟨x, hx⟩)    
\end{leancode}
\vspace{3mm}

Here we take an object $x$ from the fundamental category of $X$
and return an object from the codomain $C$.
We use \verb|filter.mem_top.mpr hX x| to show that $x \in X_1 \cup X_2$.
From this, we use \verb|set.mem_union| to obtain $x \in X_1$ or $x \in X_2$
and we can split by those cases to apply either $F_1$ or $F_2$.
We abbreviate \verb|functor_obj hX F₁ F₂| to \verb|F_obj| in our formalization to maintain clarity.
After this definition, there are two lemmas that prove that if $x \in X_k$,
then $F_k(x) = F(x)$ for $k \in \{1, 2\}$.

In the proof of \cref{Van Kampen Theorem Theorem},
$F'$ is first defined and it is then shown to be a valid definition.
Within our Lean formalization, we have to do these two parts in the reverse order.
Once we have shown that the construction is well-defined, we can define $F'$ in our formalization.
That is why \textbf{Step 2} will be completed later.

We use the definitions of \verb|covered| and \verb|covered_partwise|,
shown in \cref{Ch Directed Spaces},
to define the mapping of morphisms inductively (\textbf{Step 3}):
\vspace{3mm}
\begin{leancode}
    def functor_map_of_covered {γ : dipath x y} (hγ : covered γ hX) :
        F_obj x ⟶ F_obj y :=
        or.by_cases hγ
            (λ hγ, functor_map_aux_part_one hX h_comm hγ)
            (λ hγ, functor_map_aux_part_two hX h_comm hγ)

    def functor_map_of_covered_partwise {n : ℕ} : Π {x y : X} {γ : dipath x y}
            (hγ : covered_partwise hX γ n), F_obj x ⟶ F_obj y :=
        nat.rec_on n
          (λ x y γ hγ, F₀ hγ)
          (λ n ih x y γ hγ, (F₀ hγ.1) ≫ (ih hγ.2))
\end{leancode}
\vspace{3mm}

In \verb|functor_map_of_covered| we define what to do with a path $\gamma$ that is $1$-covered,
that is, we map it to $F_1([\gamma])$ or $F_2([\gamma])$ depending on whether
$\gamma$ is covered by $X_1$ or $X_2$.
It depends on \verb|functor_map_aux_part_one|, which specifies what $F_1([\gamma])$ should be.
We use \verb|F₀| to abbreviate \verb|functor_map_of_covered hX h_comm|.
We can then use this base case to define \verb|functor_map__of_covered_partwise|
for a $n$-covered path inductively.
If $n = 1$, we can apply \verb|F₀| by definition.
Otherwise, we are able apply \verb|F₀| to the first part of the path, which is $1$-covered,
and then \verb|functor_map_of_covered_partwise| to the second part, which is $(n-1)$-covered.
We abbreviate the construction \verb|functor_map_of_covered_partwise hX h_comm| to \verb|Fₙ|.

Since $n$ is an input for the definition, so we need to show
that it is independent of the choice for $n$.
The lemma \verb|functor_map_of_covered_partwise_unique| captures this independence (\textbf{Step 4}).

\vspace{3mm}
\begin{leancode}
    lemma functor_map_of_covered_partwise_unique {n m : ℕ} {γ : dipath x y}
    (hγ_n : covered_partwise hX γ n) (hγ_m : covered_partwise hX γ m) :
    Fₙ hγ_n = Fₙ hγ_m :=
        /- Proof omitted -/
\end{leancode}
\vspace{3mm}

This lemma makes use of the following lemma that shows that the image remains the same
if we refine the partition of $\gamma$,
so when we use a $nk$-covering instead of a $n$-covering.
\vspace{3mm} 
\begin{leancode}
    lemma functor_map_aux_of_covered_partwise_refine {n : ℕ} (k : ℕ) :
      Π {x y : X} {γ : dipath x y} (hγ_n : covered_partwise hX γ n),
      Fₙ hγ_n = Fₙ (covered_partwise_refine hX n k hγ_n) :=  
        /- Proof omitted -/
\end{leancode}
\vspace{3mm}

Now we know that the image is independent of $n$,
and because a $n > 0$ exists such that $\gamma$ is $n$-covered (shown in \verb|has_subpaths|),
we can choose one such $n$ and we obtain the following formalization of $F'$,
completing \textbf{Step 2}.
\vspace{3mm}
\begin{leancode}
    def functor_map_aux (γ : dipath x y) : F_obj x ⟶ F_obj y :=
      Fₙ (classical.some_spec (has_subpaths hX X₁_open X₂_open γ))
\end{leancode}
\vspace{3mm}

We have now formalized the $F'$ from the proof of the Van Kampen Theorem and
we first show that \cref{Fprime property 1} and \cref{Fprime property 2}
from \cref{Van Kampen Theorem Theorem} hold (\textbf{Step 5}).
\vspace{3mm}
\begin{leancode}
    lemma functor_map_aux_refl {x : X} :
      F_map_aux (dipath.refl x) = 𝟙 (F_obj x) :=
        /- Proof omitted -/

    lemma functor_map_aux_trans {x y z : X} (γ₁ : dipath x y) (γ₂ : dipath y z) :
      F_map_aux (γ₁.trans γ₂) = F_map_aux γ₁ ≫ F_map_aux γ₂ :=
        /- Proof omitted -/
\end{leancode}
\vspace{3mm}

We arrive at \textbf{Step 6} and want to show it is invariant under the \verb|dihomotopic| relation.
To do this we need to show the claim from the proof:
if we have a directed homotopy $H$ from $f$ to $g$ that is $(1, m)$-covered, then
$F'[H(\_, 1)] \circ F'[f] = F'[g] \circ F'[H(\_, 0)]$ (\textbf{Step 7}).
\vspace{3mm}
\begin{leancode}
    lemma functor_map_aux_of_homotopic_dimaps {m : ℕ} :
      Π {f g : D(I, X)} {H : directed_map.dihomotopy f g}
        (hcov : directed_map.dihomotopy.covered_partwise H hX 0 m),
        F_map_aux (dipath.of_directed_map f) ≫ F_map_aux (H.eval_at_right 1) =
        F_map_aux (H.eval_at_right 0) ≫ F_map_aux (dipath.of_directed_map g) :=
        /- Proof omitted -/
\end{leancode}
\vspace{3mm}

By using induction once again,
we end up with the lemma showing us that the choice of representative does not matter.
\vspace{3mm}
\begin{leancode}
    lemma functor_map_aux_of_dihomotopic (γ γ' : dipath x y) (h : γ.dihomotopic γ') :
        F_map_aux γ = F_map_aux γ' :=
        /- Proof omitted -/
\end{leancode}
\vspace{3mm}

We can now finally define the behavior on morphisms to obtain a functor
by using the universal mapping property of quotients.
\vspace{3mm}
\begin{leancode}
    def functor_map {x y : dπₓ X} (γ : x ⟶ y) : F_obj x ⟶ F_obj y :=
        quotient.lift_on γ F_map_aux
            (functor_map_aux_of_dihomotopic hX X₁_open X₂_open h_comm)

    ... /- Lemmas about identities and compositions -/

    def functor : (dπₓ X) ⟶ C := {
        obj := F_obj,
        map := λ x y, F_map,
        map_id' := λ x, functor_map_id hX X₁_open X₂_open h_comm x,
        map_comp' := λ x y z γ₁ γ₂, functor_map_comp hX X₁_open X₂_open h_comm γ₁ γ₂
    }
\end{leancode}
\vspace{3mm}

Here \verb|F_map| is an abbreviation for \verb|functor_map hX X₁_open X₂_open h_comm| and the
final \verb|functor| is abbreviated to simply \verb|F|.
Finally, we get to \textbf{Step 8}.
The remaining lemmas show that $F \circ \vec\Pi(j_k) = F_k$ for $k = 1$ and $k = 2$,
and that $F$ is the unique functor with this property.

\vspace{3mm}
\begin{leancode}
    lemma functor_comp_left : (dπₘ j₁) ⋙ F = F₁ := /- Proof omitted -/
    lemma functor_comp_right : (dπₘ j₂) ⋙ F = F₂ := /- Proof omitted -/
    lemma functor_uniq (F' : (dπₓ X) ⟶ C) (h₁ : (dπₘ j₁) ≫ F' = F₁)
        (h₂ : (dπₘ j₂) ≫ F' = F₂) :
        F' = F := /- Proof omitted -/
\end{leancode}
\vspace{3mm}

The Van Kampen Theorem is stated as
\vspace{3mm}
\begin{leancode}
    theorem directed_van_kampen {hX₁ : is_open X₁} {hX₂ : is_open X₂}
        {hX : X₁ ∪ X₂ = set.univ} :
        is_pushout (dπₘ i₁) (dπₘ i₂) (dπₘ j₁) (dπₘ j₂) :=
        /- Proof omitted -/  
\end{leancode}
\vspace{3mm}

The theorem \verb|directed_van_kampen| now follows easily from the lemmas above.

\subsection{Applications}\label{Van Kampen Application}
Topological spaces often have uncountably many points
and so their fundamental categories have uncountably many objects making them hard to reason about.
It is possible to reduce a fundamental category to a subset of
objects that captures the essence of the directedness.
This makes the fundamental category easier to work with and with the right selection of points,
a Van Kampen Theorem still holds.
This requires some more theory \cite{bubenik2009Models}.

Another way to keep it simple is by looking at spaces with finitely many points.
The first space we will look at is the so called \textit{discrete} or \textit{finite} unit circle.

\begin{example}\label{discrete unit circle}
    The discrete unit circle is a finite version of $S^1$
    and consists of four points $X = \{N, E, S, W\}$.
    $E$ and $W$ are the right and left point of the unit circle,
    and $N$ and $S$ represent the upper and lower arc of the circle.
    The topology is given by 
    \begin{gather*}
        \mathcal{T}_X =
         \{ \emptyset, \{ N \}, \{S \}, \{ N, S \}, \{ N, S, E \}, \{ N, S, W \}, \{ N, E, S, W\} \}
    \end{gather*}
    Let now four paths be given by 
    \begin{gather*}
        \gamma_1(t) = \begin{cases}
            E, & t = 0, \\
            N, & 0 < t.
        \end{cases},
        \gamma_2(t) = \begin{cases}
            N, & t < 1, \\
            W, & t = 1.
        \end{cases}, \\
        \gamma_3(t) = \begin{cases}
            W, & t = 0, \\
            S, & 0 < t.
        \end{cases},
        \gamma_4(t) = \begin{cases}
            S, & t < 1, \\
            E, & t = 1.
        \end{cases}.
    \end{gather*}
    Schematically we can represent the space and these paths as:
    
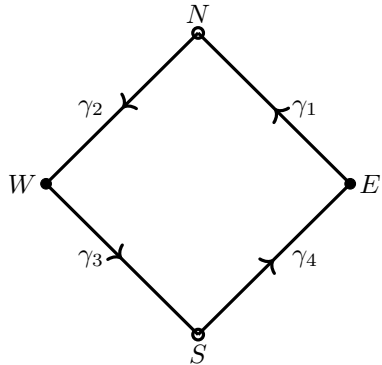
\begin{center}
        \begin{tikzpicture}
        \begin{scope}[very thick,decoration={
            markings,
            mark=at position 0.5 with {\arrow{>}}}
            ] 
            \draw[postaction={decorate}] (2,0)  -- node[right=1mm] {$\gamma_1$} (0,2);
            \draw[postaction={decorate}] (0,2)  -- node[left=1mm] {$\gamma_2$} (-2,0);
            \draw[postaction={decorate}] (-2,0) -- node[left=1mm] {$\gamma_3$} (0,-2);
            \draw[postaction={decorate}] (0,-2) -- node[right=1mm] {$\gamma_4$} (2,0);
        \end{scope}
        \filldraw[black] (-2,0) circle (2pt) node[anchor=east]{$W$};
        \filldraw[black] (2,0) circle (2pt) node[anchor=west]{$E$};
        \draw[black, very thick] (0,2) circle (2pt) node[anchor=south]{$N$};
        \draw[black, very thick] (0,-2) circle (2pt) node[anchor=north]{$S$};
        \end{tikzpicture}
    \end{center}
    We now make $X$ directed by setting $P_X$ equal to the set containing
    any monotone subparametrization of any concatenation of the four paths.
    Note that any path in $X$ is now considered directed if and only if
    it visits points in $X$ in a counterclockwise order.
    This way we mimic the directedness of $S^1_+$,
    as each directed path in that space also runs counterclockwise.

    We can split $X$ into two open subspaces $X_1 = \{N, S, E\}$ and $X_2 = \{N, S, W \}$.
    The only directed paths up to monotone reparametrization in $X_1$ are
    $\gamma_2$, $\gamma_3$, $\gamma_2 \odot \gamma_3$ and the constant paths.
    Its fundamental category has thus three objects
    with their three identities and three additional morphisms :
    $[\gamma_2], [\gamma_3]$ and $[\gamma_3] \circ [\gamma_2]$.
    Symmetrically, the fundamental category $X_2$
    has also three objects and three non-trivial morphisms.
    The intersection of $X_1$ and $X_2$ is $\{N, S\}$.
    Any path from $N$ to $S$ must go through $W$ and any path from $S$ to $N$ must go through $E$,
    so the only possible paths in $X_1 \cap X_2$ are the constant (trivial) paths.

    As $X_1$ and $X_2$ are open, the condition of \cref{Van Kampen Theorem Theorem} is satisfied.
    It follows that $X$ is a pushout of $X_1$ and $X_2$.
    We now take the category $\mathcal{C}$ containing four points $\{n, e, s, w\}$
    and whose morphisms are freely generated 
    by the four morphisms $p_1 : e \to n$, $p_2 : n \to w$, $p_3 : w \to s$ and $p_4 : s \to e$.
    Freely generated means that morphisms in $\mathcal{C}$ are exactly compositions
    of these four morphisms (or identities) and two morphisms are the same if and only
    if their sequences of these base morphisms are the same up to identities.
    We claim that $\vec{\Pi}(X)$ is isomorphic to this category $\mathcal{C}$.
    We do this by showing that $\mathcal{C}$ is also
    the pushout of $\vec{\Pi}(X_1)$ and $\vec{\Pi}(X_2)$.
    As a pushout is unique up to isomorphism \cite{leinster2016basic}, the claim will follow.

    We take two functors $j_1 : \vec\Pi(X_1) \to \mathcal{C}$
    and $j_2 : \vec\Pi(X_2) \to \mathcal{C}$.
    $j_1$ is given on the objects by $j_1(N) = n$, $j_1(S) = s$ and $j_1(E) = e$.
    The non-trivial morphisms are mapped as $j_1([\gamma_1]) = p_1$,
    $j_1([\gamma_4]) = p_4$ and $j_1([\gamma_4 \odot \gamma_1]) = p_1 \circ p_4$.
    Similarly $j_2$ is defined.
    If it clear that $j_1$ and $j_2$ agree on $\vec\Pi(X_1 \cap X_2)$,
    so we have a commutative square.

    Let $\mathcal{D}$ be another category and $F_1 : \vec\Pi(X_1) \to \mathcal{D}$
    and $F_2 : \vec\Pi(X_2) \to \mathcal{D}$ two functors that agree on $\vec\Pi(X_1 \cap X_2)$.
    We then define $F : \mathcal{C} \to \mathcal{D}$ as $F(n) = F_1(N)$, $F(e) = F_1(E)$,
    $F(s) = F_1(S)$ and $F(w) = F_2(W)$.
    On morphisms, we take $F(p_1) = F_1([\gamma_1])$, $F(p_2) = F_2([\gamma_2])$,
    $F(p_3) = F_2([\gamma_3])$ and $F(p_4) = F_1([\gamma_4])$.
    As the morphisms in $\mathcal{C}$ are freely generated by $p_1, p_2, p_3$ and $p_4$,
    this defines $F$ uniquely and it holds that $F \circ j_k = F_k$ for $k \in \{1, 2\}$,
    so $\mathcal{C}$ is indeed a pushout of $\vec\Pi(X_1)$ and $\vec\Pi(X_2)$.
    From this it follows that the morphisms in $\vec\Pi(X)$
    are also freely generated by the four morphisms
    $[\gamma_1], [\gamma_2], [\gamma_3]$ and $[\gamma_4]$.

    We can now ask the question: what do endomorphisms of $E$ in $\vec{\Pi}(X)$ look like?
    By above characterization, they are of the form
        $([\gamma_4] \circ [\gamma_3] \circ [\gamma_2] \circ [\gamma_1])^n$
    with $n \geq 0$ and each one is different.
    We now obtain an explicit isomorphism from $\vec{\pi}(X, E)$ to $(\mathbb{N}, +, 0)$ given by
    $([\gamma_4] \circ [\gamma_3] \circ [\gamma_2] \circ [\gamma_1])^n \mapsto n$.
    This supports \cref{Fundamental Monoid Unit Circle},
    where we claimed that $\vec{\pi}(S^1_+, 1) \cong (\mathbb{N}, +, 0)$.
\end{example}

\begin{example}
    We will now look at a finite version of the space considered in \cref{Counterexample Fundamental Monoid}.
    This space will be similar to the discrete unit circle we just considered.
    We take $Y = \{N, E, S, W, M\}$ with the topology given $\mathcal{T}_Y$ by
    \begin{gather*}
            \{ \emptyset, \{ N \}, \{ M \}, \{ S \}, \{ N, M \}, \{ N, S \}, \{ M, S \}, \\
                \{ N, M, S \}, \{ N, M, S, E \}, \{ N, M, S, W \}, Y \}.
    \end{gather*}
    In this case $N$, $M$ and $S$ all represent open intervals.
    Note that the discrete unit circle is a subspace of $Y$ in a natural way.
    We take the four paths $\gamma_1, ..., \gamma_4$ equal to those in the previous example
    and take two additional paths
    \begin{gather*}
        \gamma_5(t) = \begin{cases}
            E, & t = 0, \\
            M, & 0 < t.
        \end{cases},
        \gamma_6(t) = \begin{cases}
            M, & t < 1, \\
            W, & t = 1
        \end{cases}.
    \end{gather*}
    We can represent this space and these paths as:
    
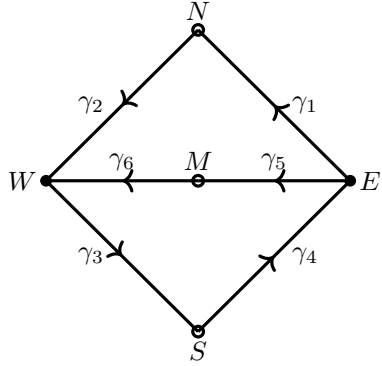
\begin{center}
        \begin{tikzpicture}
        \begin{scope}[very thick,decoration={
            markings,
            mark=at position 0.5 with {\arrow{>}}}
            ] 
            \draw[postaction={decorate}] (2,0)  -- node[right=1mm] {$\gamma_1$} (0,2);
            \draw[postaction={decorate}] (0,2)  -- node[left=1mm] {$\gamma_2$} (-2,0);
            \draw[postaction={decorate}] (-2,0) -- node[left=1mm] {$\gamma_3$} (0,-2);
            \draw[postaction={decorate}] (0,-2) -- node[right=1mm] {$\gamma_4$} (2,0);
            \draw[postaction={decorate}] (2,0)  -- node[above] {$\gamma_5$} (0,0);
            \draw[postaction={decorate}] (0,0)  -- node[above] {$\gamma_6$} (-2,0);
        \end{scope}
        \filldraw[black] (-2,0) circle (2pt) node[anchor=east]{$W$};
        \filldraw[black] (2,0) circle (2pt) node[anchor=west]{$E$};
        \draw[black, very thick] (0,2) circle (2pt) node[anchor=south]{$N$};
        \draw[black, very thick] (0,-2) circle (2pt) node[anchor=north]{$S$};
        \draw[black, very thick] (0,0) circle (2pt) node[anchor=south]{$M$};
        \end{tikzpicture}
    \end{center}
    We take the open subspaces $Y_1 = \{ N, M, S, E \}$ and $Y_2 = \{ N, M, S, W \}$.
    The non-trivial morphisms in $\vec{\Pi}(Y_1)$ are $[\gamma_1], [\gamma_4],
        [\gamma_5], [\gamma_5] \circ [\gamma_4]$ and $[\gamma_1] \circ [\gamma_5]$. 
    In $\vec{\Pi}(Y_2)$ they are $[\gamma_2], [\gamma_3], [\gamma_6], 
        [\gamma_3] \circ [\gamma_2]$ and $[\gamma_3] \circ [\gamma_6]$.
    $\vec{\Pi}(Y_1 \cap Y_2)$ only consists of three points
    $N, M$ and $S$ together with their identities.
    Just like in the previous example,
    \cref{Van Kampen Theorem Theorem} tells us that the morphisms in $\vec{\Pi}(Y)$ are
    freely generated by the six morphisms $[\gamma_i]$, $1 \le i \le 6$.

    This time we are interested in the non-trivial endomorphisms of $W$.
    If we start in point $W$,
    first the morphism $[\gamma_4] \circ [\gamma_3]$ must be followed to reach point $E$.
    From there we have a choice to return to $W$:
    either $[\gamma_2] \circ [\gamma_1]$ or $[\gamma_6] \circ [\gamma_5]$.
    We find that endomorphisms of $W$ are sequences of two loops
    $[\gamma_2] \circ [\gamma_1] \circ [\gamma_4] \circ [\gamma_3]$ and 
    $[\gamma_6] \circ [\gamma_5] \circ [\gamma_4] \circ [\gamma_3]$.
    This structure is isomorphic to $\mathbb{N} * \mathbb{N}$,
    agreeing with \cref{Counterexample Fundamental Monoid}.
\end{example}

\section{Conclusion and Further Research}\label{Ch Conclusion}
 
In this article, we presented important concepts from directed topology.
These allowed us to state and prove a directed version of the Van Kampen Theorem.
If its simple conditions are satisfied,
it allows us to calculate the fundamental category of a directed space
using the fundamental categories of subspaces.
We showcased how we formalized that theorem
and the theory leading up to it using the Lean proof assistant. \medbreak

Our formalization can be extended to other concepts and theorems from directed topology.
For example, Bubenik's approach of restricting the fundamental category to a
(finite) full subcategory is a clear extension of the theory we have formalized in Lean \cite{bubenik2009Models}.

At the moment, MathLib does not have a version of the Van Kampen Theorem for groupoids,
originally proven by Brown in 1968 \cite{brown1968Topology,brown2006Topology}.
As Grandis based his proof of the directed version on Brown's version,
our formalization can conversely be adapted to suit the undirected case.
The undirected version could also be shown to be a corollary of the directed version.
This is because the fundamental category of a topological space equipped
with the maximal directedness coincides with the fundamental groupoid of the topological space.

It might also prove interesting to further investigate sufficient and necessary preconditions
for the Van Kampen Theorem, as stated in \cref{Van Kampen Theorem}.
As the proof of the Van Kampen Theorem suggests,
we want to be able to cover both paths and homotopies,
but whether that is truly necessary is something that future research will have to tell.

\emergencystretch=1em
\printbibliography{}

\end{document}